\documentclass[11pt]{article}

\pdfoutput=1 

\usepackage[top=1in, bottom=1.45in, left=1in, right=1in]{geometry}
\usepackage{lscape}
\usepackage{amsmath}
\usepackage{amsthm}
\usepackage{subfig}
\usepackage{setspace}
\usepackage{amssymb}
\usepackage{color}
\usepackage{soul}
\usepackage{epstopdf}
\usepackage{dsfont}
\usepackage{multicol}
\usepackage{lipsum}
\usepackage{algorithm}
\usepackage{latexsym}
\usepackage{amsfonts,amscd}
\usepackage{upgreek}
\usepackage{esint}
\usepackage[round,authoryear]{natbib}
\usepackage{enumitem}
\usepackage{caption} 
\usepackage{url}
\usepackage{authblk}
\usepackage{graphicx}

\newtheorem{definition}{Definition}
\newtheorem{lemma}{Lemma}
\newtheorem{theorem}{Theorem}
\newtheorem{rem}{Remark}
\let\oldproofname=\proofname
\renewcommand{\proofname}{\rm\bf{\oldproofname}}

\newcommand{\bN}{{\mathbb N}}

\newcommand{\bR}{{\mathbb R}}

\newcommand{\im}{\mathsf{im}}
\renewcommand{\ker}{\mathsf{ker}}
\newcommand{\cD}{\mathcal{D}}

\begin{document}

\title{Persistence Images: A Stable Vector Representation of Persistent Homology}

\author[1]{Henry Adams}
    
\author[2]{Sofya Chepushtanova}
	
\author[1]{Tegan Emerson}
	
\author[3]{Eric Hanson}
	
\author[1]{Michael Kirby}
	
\author[4]{Francis Motta} 
	
\author[1]{Rachel Neville} 
	
\author[1]{Chris Peterson} 
	
\author[1]{Patrick Shipman} 
	
\author[5]{Lori	Ziegelmeier} 

 \affil[1]{Department of Mathematics, Colorado State University}
 \affil[2]{Department of Mathematics and Computer Science, Wilkes University}
 \affil[3]{Department of Mathematics, Texas Christian University}
 \affil[4]{Department of Mathematics, Duke University}
 \affil[5]{Department of Mathematics, Statistics, and Computer Science, Macalester College}

\date{} 

\maketitle

\begin{abstract}
Many datasets can be viewed as a noisy sampling of an underlying space, and tools from topological data analysis can characterize this structure for the purpose of knowledge discovery. One such tool is persistent homology, which provides a multiscale description of the homological features within a dataset. A useful representation of this homological information is a \emph{persistence diagram} (PD). 
Efforts have been made to map PDs into spaces with additional structure valuable to machine learning tasks. 
We convert a PD to a finite-dimensional vector representation which we call a \emph{persistence image} (PI), and prove the stability of this transformation with respect to small perturbations in the inputs. The discriminatory power of PIs is compared against existing methods, showing significant performance gains.
We explore the use of PIs with vector-based machine learning tools, such as linear sparse support vector machines, which identify features containing discriminating topological information.  Finally, high accuracy inference of parameter values from the dynamic output of a discrete dynamical system (the \emph{linked twist map}) and a partial differential equation (the \emph{anisotropic Kuramoto-Sivashinsky equation}) provide a novel application of the discriminatory power of PIs.
\vskip1ex
\noindent \textbf{Keywords:} topological data analysis, persistent homology, persistence images, machine learning, dynamical systems
\end{abstract}


\section{Introduction}\label{sec:intro}
    
In recent years, the field of topology has grown to include a large set of computational tools \citep{Edelsbrunner10}. One of the fundamental tools is persistent homology, which tracks how topological features appear and disappear in a nested sequence of topological spaces \citep{edelsbrunner2008persistent,computingPH}. This multiscale information can be represented as a \emph{persistence diagram} (PD), a collection of points in the plane where each point $(x,y)$ corresponds to a topological feature that appears at scale $x$ and disappears at scale $y$.  We say the feature has a persistence value of $y-x$. This compact summary of topological characteristics by finite multi-sets of points in the plane is responsible, in part, for the surge of interest in applying persistent homology to the analysis of complex, often high-dimensional data. Computational topology has been successfully applied to a broad range of data-driven disciplines~\citep{windowsandpersistence, hippocampalPH, chung2009persistence, imagewebs, visionTDA, Swarms, ions}. 

Concurrent with this revolution in computational topology, a growing general interest in data analysis has driven advances in data mining, pattern recognition, and machine learning (ML). Since the space of PDs can be equipped with a metric structure (\emph{bottleneck} or \emph{Wasserstein} \citep{probabilityonPD, turner2014frechet}), and since these metrics reveal the stability of PDs under small perturbations of the data they summarize \citep{stabilityPD, cohen2010lipschitz, chazal2014persistence}, it is possible to perform a variety of ML techniques using PDs as a statistic for clustering data sets. However, many other useful ML tools and techniques (e.g.\ support vector machines (SVM), decision tree classification, neural networks, feature selection, and dimension reduction methods) require more than a metric structure. In addition, the cost of computing the bottleneck or Wasserstein distance grows quickly as the number of off-diagonal points in the diagrams increases \citep{di2015comparing}. To resolve these issues, considerable effort has been made to map PDs into spaces which are suitable for other ML tools \citep{bubenik2015statistical, reininghaus2015stable,bendich2014topologicalLiz,adcock2012ring,donatini1998size,ferri1997point,chung2009persistence,pachauri2011topology,bendich2014persistent,chen2015statistical,carriere2015stable,di2015comparing}. Each approach has benefits and drawbacks, which we review in \S\ref{sec:relatedwork}.  With these in mind, we are led to pose the following question:

\vspace{0.5em}
\noindent\textbf{Problem Statement:} How can we represent a persistence diagram so that:
\begin{enumerate}[label=(\roman*)] 
\item the output of the representation is a vector in $\bR^n$,
\item the representation is stable with respect to input noise,
\item the representation is efficient to compute,
\item the representation maintains an interpretable connection to the original PD, and
\item\label{item:relativeImportance} the representation allows one to adjust the relative importance of points in different regions of the PD?
\end{enumerate}


The main contribution of this paper is to provide a finite-dimensional-vector representation of a PD called a \emph{persistence image} (PI). We first map a persistence diagram $B$ to an integrable function $\rho_B\colon\bR^2 \rightarrow \bR$ called a \emph{persistence surface}. The surface $\rho_B$ is defined as a weighted sum of Gaussian functions\footnote{Or more generally, a weighted sum of probability density functions}, one centered at each point in the PD. The idea of persistence surfaces has appeared even prior to the development of persistent homology, in \citet{donatini1998size} and \cite{ferri1997point}. Taking a discretization of a subdomain of $\rho_B$ defines a grid.  A persistence image, i.e.\ a matrix of pixel values, can be created by computing the integral of $\rho_B$ on each grid box. This PI is a ``vectorization'' of the PD, and provides a solution to the problem statement above. 

Criteria (i) is our primary motivation for developing PIs.  A large suite of ML techniques and statistical tools (means and variances) already exist to work with data in $\mathbb{R}^n$.  Additionally, such a representation allows for the use of various distance metrics ($p$-norms and angle based metrics) and other measures of (dis)similarity.  The remaining criteria of the problem statement (ii-v) further ensure the usefulness of this representation.  

The desired flexibility of \ref{item:relativeImportance} is accomplished by allowing one to build a PI as a weighted sum of Gaussians, where the weightings may be chosen from a broad class of weighting functions.\footnote{Weighting functions are restricted only to the extent necessary for our stability results in \S\ref{sec:stability}.} For example, a typical interpretation is that points in a PD of high persistence are more important than points of low persistence (which may correspond to noise). One may therefore build a PI as a weighted sum of Gaussians where the weighting function is non-decreasing with respect to the persistence value of each PD point. However, there are situations in which one may prefer different measures of importance. Indeed, \citet{bendich2014persistent} find that, in their regression task of identifying a human brain's age from its arterial geometry, it is the points of medium persistence (not high persistence) that best distinguish the data. In such a setting, one may choose a weighting function with largest values for the points of medium persistence. In addition, the Homology Inference Theorem \citep{stabilityPD} states that when given a sufficiently dense finite sample from a space $X$, it is the points in the PD with sufficiently small birth times (and sufficiently high persistence) which recover the homology groups of the space; hence one may choose a weighting function that emphasizes points near the death-axis and away from the diagonal, as indicated in the leftmost yellow rectangle of \citet[Figure~2.4]{bendichThesis}. A potential disadvantage of the flexibility in \ref{item:relativeImportance} is that it requires a choice; however, prior knowledge of one's particular problem may inform that choice.  Moreoever, our examples illustrate the effectiveness of a standard choice of weighting function that is non-decreasing with the persistence value.
   
The remainder of this article is organized as follows. Related work connecting topological data analysis and ML is reviewed in \S\ref{sec:relatedwork}, and \S\ref{sec:background} gives a brief introduction to persistent homology, PDs from point cloud data, PDs from functions, and the bottleneck and Wasserstein metrics. PIs are defined in \S\ref{sec:persistenceimages} and their stability with respect to the 1-Wasserstein distance between PDs is proved in \S\ref{sec:stability}.  Lastly,  \S\ref{sec:experiments} contains examples of ML techniques applied to PIs generated from samples of common topological spaces, an applied dynamical system modeling turbulent mixing, and a partial differential equation describing pattern formation in extended systems driven far from equilibrium. Our code for producing PIs is publicly available at \url{https://github.com/CSU-TDA/PersistenceImages}.

\section{Related Work}\label{sec:relatedwork}
The space of PDs can be equipped with the bottleneck or Wasserstein metric (defined in \S\ref{sec:background}), and one reason for the popularity of PDs is that these metrics are stable with respect to small deviations in the inputs \citep{stabilityPD, cohen2010lipschitz, chazal2014persistence}. Furthermore, the bottleneck metric  allows one to define Fr{\'e}chet means and variances for a collection of PDs \citep{probabilityonPD, turner2014frechet}. However, the structure of a metric space alone is insufficient for many ML techniques, and a recent area of interest in the topological data analysis community has been encoding PDs in ways that broaden the applicability of persistence. For example, \citet{adcock2012ring} study a ring of algebraic functions on the space of persistence diagrams, and \citet{verovsek2016tropical} identifies tropical coordinates on the space of diagrams. \citet{ferri1999representing} and \citet{di2015comparing} encode a PD using the coefficients of a complex polynomial that has the points of the PD as its roots. 

\citet{bubenik2015statistical} develops the notion of a persistence landscape, a stable functional representation of a PD that lies in a Banach space. A persistence landscape (PL) is a function $\lambda\colon\bN\times\bR\to[-\infty,\infty]$, which can equivalently be thought of as a sequence of functions $\lambda_k\colon\bR\to[-\infty,\infty]$. For $1\le p\le\infty$ the $p$-landscape distance between two landscapes $\lambda$ and $\lambda'$ is defined as $\|\lambda-\lambda'\|_p$; the $\infty$-landscape distance is stable with respect to the bottleneck distance on PDs, and the $p$-landscape distance is continuous with respect to the $p$-Wasserstein distance on PDs. One of the motivations for defining persistence landscapes is that even though Fr{\'e}chet means of PDs are not necessarily unique \citep{probabilityonPD}, a set of persistence landscapes does have a unique mean.  Unique means are also a feature of PIs as they are vector representations. An advantage of PLs  over PIs is that the map from a PD to a PL is easily invertible; an advantage of PIs over PLs is that PIs live in Euclidean space and hence are amenable to a broader range of ML techniques. In \S\ref{sec:experiments}, we compare PDs, PLs, and PIs in a classification task on synthetic data sampled from common topological spaces. We find that PIs behave comparably or better than PDs when using ML techniques available to both representations, but PIs are significantly more efficient to compute. Also, PIs outperform PLs in the majority of the classification tasks and are of comparable computational efficiency.

A vector representation of a PD, due to \citet{carriere2015stable}, can be obtained by rearranging the entries of the distance matrix between points in a PD. In their Theorem~3.2, they prove that both the $L^\infty$ and $L^2$ norms between their resulting vectors are stable with respect to the bottleneck distance on PDs. They remark that while the $L^\infty$ norm is useful for nearest-neighbor classifiers, the $L^2$ norm allows for more elaborate algorithms such as SVM. However, though their stability result for the $L^\infty$ norm is well-behaved, their constant for the $L^2$ norm scales undesirably with the number of points in the PD. We provide this as motivation for our Theorem~\ref{thm:image-stable-Gaussians}, in which we prove the $L^\infty$, $L^1$, and $L^2$ norms for PI vectors are stable with respect to the 1-Wasserstein distance between PDs, and in which none of the constants depend on the number of points in the PD.

By superimposing a grid over a PD and counting the number of topological features in each bin, \citet{bendich2014topologicalLiz} create a feature vector representation. An advantage of this approach is that the output is easier to interpret than other more complicated representations, but a disadvantage is that the vectors are not stable for two reasons:
\begin{itemize}
\item[(i)] an arbitrarily small movement of a point in a PD may move it to another bin, and
\item[(ii)] a PD point emerging from the diagonal creates a discontinuous change.
\end{itemize}
Source (i) of instability can be improved by first smoothing a PD into a surface. This idea has appeared multiple times in various forms --- even prior to the development of persistent homology, \citet{donatini1998size} and \citet{ferri1997point} convert size functions (closely related to 0-dimensional PDs) into surfaces by taking a sum of Gaussians centered on each point in the diagram. This conversion is not stable due to (ii), and we view our work as a continued study of these surfaces, now also in higher homological dimensions, in which we introduce a weighting function\footnote{Our weighting function is continuous and zero for points of zero persistence, i.e.\ points along the diagonal.} to address (ii) and obtain stability. \citet{chung2009persistence} produce a surface by convolving a PD with the characteristic function of a disk, and \citet{pachauri2011topology} produce a surface by centering a Gaussian on each point, but both of these methods lack stability again due to (ii). Surfaces produced from random PDs are related to the empirical intensity plots of \citet{edelsbrunner2012current}.

\citet{reininghaus2015stable} produce a stable surface from a PD by taking the sum of a positive Gaussian centered on each PD point together with a negative Gaussian centered on its reflection below the diagonal; the resulting surface is zero along the diagonal. This approach is similar to ours, and indeed we use \citep[Theorem~3]{reininghaus2015stable} to show that our persistence surfaces are stable only with respect to the 1-Wasserstein distance (Remark~\ref{rem:only-1-stable}). Nevertheless, we propose our independently-developed surfaces as an alternative stable representation of PDs with the following potential advantages. First, our sum of non-negatively weighted Gaussians may be easier to interpret than a sum including negative Gaussians. Second, we produce vectors from our surfaces with well-behaved stability bounds, allowing one to use vector-based learning methods such as linear SVM. Indeed, \citet{zeppelzauer2016topological} report that while the kernel of \citet{reininghaus2015stable} can be used with nonlinear SVMs, in practice this becomes inefficient for a large number of training vectors because the entire kernel matrix must be computed. Third, while the surface of \citet{reininghaus2015stable} weights persistence points further from the diagonal more heavily, there are situations in which one may prefer different weightings, as discussed in \S\ref{sec:intro} and item  \ref{item:relativeImportance} of our Problem Statement. Hence, one may want weightings on PD points that are non-increasing or even decreasing when moving away from the diagonal, an option available in our approach.

We produce a persistence surface from a PD by taking a weighted sum of Gaussians centered at each point. We create vectors, or PIs, by integrating our surfaces over a grid, allowing ML techniques for finite-dimensional vector spaces to be applied to PDs. Our PIs are stable, and distinct homology dimensions may be concatenated together into a single vector to be analyzed simultaneously. Our surfaces are studied from the statistical point of view by \citet{chen2015statistical}, who cite a preprint version of our work; their applications in Section 4 use the $L^1$ norm between these surfaces, which can be justified as a reasonable notion of distance due to our Theorem~\ref{thm:surface-stable-Gaussians} that proves the $L^1$ distance between such surfaces is stable.

\citet{zeppelzauer2016topological} apply persistent images to 3D surface analysis for archeological data, in which the machine learning task is to distinguish scans of natural rock surfaces from those containing ancient human-made engravings. The authors state they select PIs over other topological methods because PIs are computationally efficient and can be used with a broader set of ML techniques. PIs are compared to an aggregate topological descriptor for a persistence diagram: the first entry of this vector is the number of points in the diagram, and the remaining entries are the minimum, maximum, mean, standard deviation, variance, 1st-quartile, median, 3rd-quartile, sum of square roots, sum, and sum of squares of all the persistence values. In their three experiments, the authors find the following.
\begin{itemize}
\item When classifying natural rock surfaces from engravings using persistent diagrams produced from the sublevel set filtration, PIs outperform the aggregate descriptor.
\item When the natural rock and engraved surfaces are first preprocessed using the completed local binary pattern (CLBP) operator for texture classificiation \citep{guo2010completed}, PIs outperform the aggregate descriptor.
\item The authors added PIs and the aggregate descriptor to eleven different non-topological baseline descriptors, and found that the classification accuracy of the baseline descriptor was improved more by the addition of PIs than by the addition of the aggregate descriptor.
\end{itemize}
Furthermore, Table~1 of \citet{zeppelzauer2016topological} demonstrates that for their machine learning task, PIs have low sensitivity to the parameter choices of resolution and variance (\S\ref{sec:persistenceimages}).

\section{Background on Persistent Homology}\label{sec:background}

Homology is an algebraic topological invariant that, roughly speaking, describes the holes in a space. The $k$-dimensional holes (connected components, loops, trapped volumes, etc.) of a topological space $X$ are encoded in an algebraic structure called the $k$-th homology group of $X$, denoted $H_k(X)$.  The rank of this group is referred to as the \emph{$k$-th Betti number}, $\beta_k$, and counts the number of independent $k$-dimensional holes. For a comprehensive study of homology, see \citet{hatcher2002algebraic}.

Given a nested sequence of topological spaces $X_1 \subseteq X_2 \subseteq \ldots \subseteq X_n$, the inclusion $X_i \subseteq X_{i'}$ for $i\le i'$ induces a linear map $H_k(X_i) \rightarrow H_k(X_{i'})$ on the corresponding $k$-th homology for all $k\ge0$. The idea of \emph{persistent homology} is to track elements of $H_k(X_i)$ as the scale (or ``time'') parameter $i$ increases \citep{edelsbrunner2008persistent,computingPH,Edelsbrunner10}. A standard way to represent persistent homology information is a \emph{persistence diagram} (PD)\footnote{Another standard representation is the barcode \citep{barcodes}.}, which is a multiset of points in the Cartesian plane $\bR^2$. For a fixed choice of homological dimension $k$, each homological feature is represented by a point $(x,y)$, whose \emph{birth} and \emph{death} indices $x$ and $y$ are the scale parameters at which that feature first appears and disappears, respectively. Since all topological features die after they are born, necessarily each point appears on or above the diagonal line $y=x$. A PD is a multiset of such points, as distinct topological features may have the same birth and death coordinates.\footnote{By convention, all points on the diagonal are taken with infinite multiplicity. This facilitates the definitions of the $p$-Wasserstein and bottleneck distances below.} Points near the diagonal are often considered to be noise while those further from the diagonal represent more robust topological features. 

In this paper, we produce PDs from two different types of input data:
\begin{enumerate}[label=(\roman*)]
\item When our data is a a point cloud, i.e.\ a finite set of points in some space, then we produce PDs using Vietoris--Rips filtration.
\item When our data is a real-valued function, then we produce PDs using the sublevel set filtration.\footnote{As explained in \S\ref{app:PD_Functions}, (i) can be viewed as a special case of (ii).}
\end{enumerate}
For setting (i), point cloud data often comes equipped with a metric or a measure of internal similarity and is rich with latent geometric content. One approach to identifying geometric shapes in data is to consider the dataset as the vertices of a simplicial complex and to add edges, triangles, tetrahedra, and higher-dimensional simplices whenever their diameter is less than a fixed choice of scale. This topological space is called the Vietoris--Rips simplicial complex, which we introduce in more detail in \S\ref{PD_Data}. The homology of the Vietoris--Rips complex depends crucially on the choice of scale, but persistent homology eliminates the need for this choice by computing homology over a range of scales \citep{carlsson2009topology,barcodes}. In \S\ref{sec:kmedoids}--\ref{sec:linkedtwist}, we obtain PDs from point cloud data using the Vietoris--Rips filtered simplicial complex, and we use ML techniques to classify the point clouds by their topological features.

In setting (ii), our input is a real valued function $f\colon X \to\bR$ defined on some domain $X$. One way to understand the behavior of map $f$ is to understand the topology of its sublevel sets $f^{-1}((-\infty,\epsilon])$. By letting $\epsilon$ increase, we obtain an increasing sequence of topological spaces, called the sublevel set filtration, which we introduce in more detail in \S\ref{app:PD_Functions}. In \S\ref{sec:aks}, we obtain PDs from surfaces $u\colon[0,1]^2\to\bR$ produced from the Kuramoto-Sivashinsky equation, and we use ML techniques to perform parameter classification.

In both settings the output of the persistent homology computation is a collection of PDs encoding homological features of the data across a range of scales.  Let $\cD$ denote the set of all PDs. The space $\cD$ can be endowed with metrics as studied by \citet{stabilityPD} and \citet{probabilityonPD}. The \emph{$p$-Wasserstein distance} defined between two PDs $B$ and $B'$ is given by
\[W_p(B, B') = \underset{\gamma: B \rightarrow B'}{\inf} \Big( \underset{u \in B}{\sum} ||u-\gamma(u)||_{\infty}^p \Big) ^{1/p},\]
where $1 \leq p < \infty$ and $\gamma$ ranges over bijections between $B$ and $B'$. Another standard choice of distance between diagrams is $\displaystyle W_{\infty}(B, B') = \underset{\gamma: B \rightarrow B'}{\inf}\underset{u \in B}{\sup} ||u- \gamma(u)||_{\infty},$
referred to as the \emph{bottleneck distance}.  These metrics allow us to measure the (dis)similarity between the homological characteristics of two datasets.

\section{Persistence Images}\label{sec:persistenceimages}

We propose a method for converting a PD into a vector while maintaining an interpretable connection to the original PD. Figure~\ref{fig:pipeline} illustrates the pipeline from data to PI starting with spectral and spatial information in $\mathbb{R}^5$ from an immunofluorescent image of a circulating tumor cell \citep{emersonCTC}.

Precisely, let $B$ be a PD in birth-death coordinates\footnote{We omit points that correspond to features with infinite persistence, e.g.\ the $H_0$ feature corresponding to the connectedness of the complete simplicial complex.}. Let $T \colon \bR^2 \to \bR^2$ be the linear transformation $T(x,y)=(x,y-x)$, and let $T(B)$ be the transformed multiset in birth-persistence coordinates\footnote{Instead of birth-persistence coordinates, one could also use other choices such as birth-death or (average size)-persistence coordinates. Our stability results (\S\ref{sec:stability}) still hold with only a slight modification to the constants.}, where each point $(x,y)\in B$ corresponds to a point $(x,y-x)\in T(B)$. Let $\phi_u\colon\bR^2\to\bR$ be a differentiable probability distribution with mean $u=(u_x,u_y)\in\bR^2$.  In all of our applications, we choose this distribution to be the normalized symmetric Gaussian $\phi_u=g_u$ with mean $u$ and variance $\sigma^2$ defined as

$$\displaystyle g_u(x,y) = \dfrac{1}{2\pi \sigma^2} e^{-[(x-u_x)^2+(y-u_y)^2]/2\sigma^2}\text{.}$$ 

\noindent Fix a nonnegative weighting function $f\colon \bR^2 \to \bR$ that is zero along the horizontal axis, continuous, and piecewise differentiable. With these ingredients we transform the PD into a scalar function over the plane.

\begin{definition}
For $B$ a PD, the corresponding \emph{persistence surface} $\rho_B \colon \bR^2 \to \bR$ is the function  
\[\displaystyle\rho_B(z) = \sum_{u\in T(B)} f(u)\phi_u(z).\]
\end{definition}
The weighting function $f$ is critical to ensure the transformation from a PD to a persistence surface is stable, which we prove in \S\ref{sec:stability}.

Finally, the surface $\rho_B(z)$ is reduced to a finite-dimensional vector by discretizing a relevant subdomain and integrating $\rho_B(z)$ over each region in the discretization. In particular, we fix a grid in the plane with $n$ boxes (pixels) and assign to each the integral of $\rho_B$ over that region. 

\begin{definition}
For $B$ a PD, its \emph{persistence image} is the collection of pixels $I(\rho_B)_p=\iint_p \rho_B\;dydx$.
\end{definition}

\begin{figure}
\begin{centering}
\includegraphics[width=1.0\textwidth]{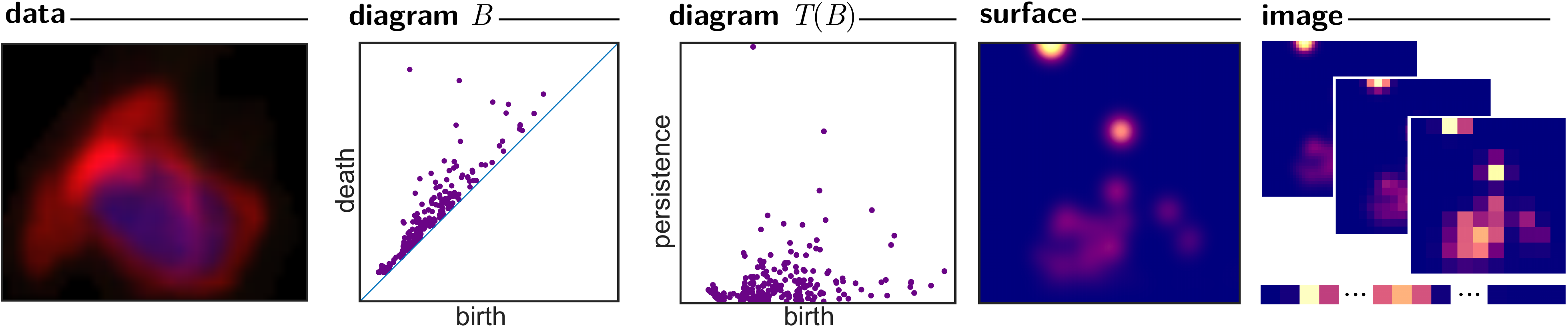}
\caption{Algorithm pipeline to transform data into a persistence image.}
\label{fig:pipeline}
\end{centering}
\end{figure}

PIs provide a convenient way to combine PDs of different homological dimensions into a single object. Indeed, suppose in an experiment the PDs for $H_0$, $H_1$, \ldots, $H_k$ are computed. One can concatenate the PI vectors for $H_0$, $H_1$, \ldots, $H_k$ into a single vector representing all homological dimensions simultaneously, and then use this concatenated vector as input into ML algorithms.  

There are three choices the user makes when generating a PI: the resolution, the distribution  (and its associated parameters), and the weighting function. 

\textbf{Resolution of the image:} The resolution of the PI corresponds to the grid being overlaid on the PD. The classification accuracy in the PI framework appears to be fairly robust to choice of resolution, as discussed in \S\ref{subsec:param_selection} and \citep{zeppelzauer2016topological}.

\textbf{The Distribution:}
Our method requires the choice of a probability distribution which is associated to each of the points in the PD. The examples in this paper use a Gaussian centered at each point, but other distributions may be used. The Gaussian distribution depends on a choice of variance: we leave this choice as an open problem, though the experiments in \S\ref{subsec:param_selection} and \citep{zeppelzauer2016topological} show a low sensitivity to the choice of variance.

\textbf{The Weighting Function:} In order for our stability results in \S\ref{sec:stability} to hold, our weighting function $f\colon \bR^2 \to \bR$ must be zero along the horizontal axis (the analogue of the diagonal in birth-persistence coordinates), continuous, and piecewise differentiable. A simple choice is a weighting function that depends only on the vertical persistence coordinate $y$. In order to weight points of higher persistence more heavily, functions which are nondecreasing in $y$, such as sigmoidal functions, are a natural choice. However, in certain applications such as \citet{bendich2014persistent} it may be points of small or medium persistence that perform best for ML tasks, and hence, one may choose to use more general weighting functions. In our experiments in \S\ref{sec:experiments}, we use a piecewise linear weighting function $f\colon \bR^2 \to \bR$ which only depends on the persistence coordinate $y$. Given $b>0$, define $w_b\colon \bR \to \bR$ via
\[w_b(t)=\begin{cases}
0 &\mbox{if }t\le0,\\
\frac{t}{b} &\mbox{if }0<t<b\mbox{, and}\\
1 &\mbox{if }t\ge b.
\end{cases}\]
We use $f(x,y)=w_b(y)$, where $b$ is the persistence value of the most persistent feature in all trials of the experiment. 

In the event that the birth coordinate is zero for all points in the PD, as is often the case for $H_{0}$, it is possible to generate a 1-dimensional (instead of 2-dimensional) PI using 1-dimensional distributions. This is the approach we adopt. Appendix \ref{app:PIEx} displays examples of PIs for the common topological spaces of a circle and a torus with various parameter choices.

\section{Stability of Persistence Surfaces and Images}\label{sec:stability}

Due to the unavoidable presence of noise or measurement error, tools for data analysis ought to be stable with respect to small perturbations of the inputs. Indeed, one reason for the popularity of PDs in topological data analysis is that the transformation of a data set to a PD is stable (Lipschitz) with respect to the bottleneck metric and -- given some mild assumptions about the underlying data -- is also stable with respect to the Wasserstein metrics (\cite{Edelsbrunner10}). In \S\ref{subsec:stability_general}, we show that persistence surfaces and images are stable with respect to the 1-Wasserstein distance between PDs. In \S\ref{subsec:stability_Gaussian}, we prove stability with improved constants when the PI is constructed using the Gaussian distribution.

\subsection{Stability for general distributions}\label{subsec:stability_general}

For $h\colon\bR^2\to\bR$ differentiable, define $|\nabla h|=\sup_{z\in\bR^2}\|\nabla h(z)\|_2$ to be the maximal norm of the gradient vector of $h$, i.e.\ the largest directional derivative of $h$. It follows by the fundamental theorem of calculus for line integrals that for all $u,v\in\bR^2$, we have
\begin{equation}\label{eq:directional-derivative}
|h(u)-h(v)|\le |\nabla h|\; \|u-v\|_2.
\end{equation}

We may safely denote $|\nabla \phi_u|$ by $|\nabla \phi|$ and $\|\phi_u\|_\infty$ by $\|\phi\|_\infty$ since the maximal directional derivative and supremum of a fixed differentiable probability distribution are invariant under translation. Note that
\begin{equation}\label{eq:directional-derivative-gaussian}
\|\phi_u-\phi_v\|_\infty \le |\nabla \phi|\; \|u-v\|_2
\end{equation}
since for any $z\in\bR^2$ we have $|\phi_u(z)-\phi_v(z)| = |\phi_u(z)-\phi_u(z+u-v)| \le |\nabla \phi|\; \|u-v\|_2$.

Recall that our nonnegative weighting function $f\colon \bR^2 \to \bR$ is defined to be zero along the horizontal axis, continuous, and piecewise differentiable.

\begin{lemma}\label{lem:L-infty-rho}
For $u,v\in\bR^2$, we have $\|f(u)\phi_u-f(v)\phi_v\|_\infty \le \bigl(\|f\|_\infty |\nabla \phi| + \|\phi\|_\infty |\nabla f|\bigr)\|u-v\|_2$.
\end{lemma}

\begin{proof}
For any $z\in\bR^2$, we have
\begin{align*}
|f(u)\phi_u(z)-f(v)\phi_v(z)| &= \bigl|f(u)\bigl(\phi_u(z)-\phi_v(z)\bigl) + \bigl(f(u)-f(v)\bigl)\phi_v(z)\bigr| \\
&\le \|f\|_\infty\; |\phi_u(z)-\phi_v(z)| + \|\phi\|_\infty |f(u)-f(v)| \\
&\le \|f\|_\infty |\nabla \phi|\; \|u-v\|_2 + \|\phi\|_\infty |\nabla f|\; \|u-v\|_2 &&\mbox{by \eqref{eq:directional-derivative-gaussian} and \eqref{eq:directional-derivative}}\\
&= \bigl(\|f\|_\infty |\nabla \phi| + \|\phi\|_\infty |\nabla f|\bigr)\|u-v\|_2.
\end{align*}
\end{proof}

\begin{theorem}\label{thm:surface-stable-general}
The persistence surface $\rho$ is stable with respect to the 1-Wasserstein distance between diagrams: for $B,B'\in\cD$ we have
\[ \| \rho_B - \rho_{B'}\|_\infty \leq \sqrt{10} \bigl(\|f\|_\infty |\nabla \phi| + \|\phi\|_\infty |\nabla f|\bigr) W_1(B,B'). \]
\end{theorem}

\begin{proof}
Since we assume $B$ and $B'$ consist of finitely many points, there exists a matching $\gamma$ that achieves the infimum in the Wasserstein distance. Then 
\begin{align*}
\| \rho_B - \rho_{B'} \|_\infty & = \|\sum_{u\in T(B)} f(u)\phi_u - \sum_{u\in T(B)} f(\gamma(u))\phi_{\gamma(u)}\|_\infty \\
&\leq \sum_{u\in T(B)} \|f(u)\phi_u - f(\gamma(u))\phi_{\gamma(u)} \|_\infty \\
&\leq \bigl(\|f\|_\infty |\nabla \phi| + \|\phi\|_\infty |\nabla f|\bigr) \sum_{u\in T(B)} \|u-\gamma(u)\|_2 &&\mbox{by Lemma~\ref{lem:L-infty-rho}.} \\
&\leq \sqrt{2} \bigl(\|f\|_\infty |\nabla \phi| + \|\phi\|_\infty |\nabla f|\bigr) \sum_{u\in T(B)} \|u-\gamma(u)\|_\infty &&\mbox{since }\|\cdot\|_2 \leq \sqrt{2} \|\cdot\|_\infty\mbox{ in }\bR^2\\
&\leq \sqrt{10} \bigl(\|f\|_\infty |\nabla \phi| + \|\phi\|_\infty |\nabla f|\bigr) \sum_{u\in B} \|u-\gamma(u)\|_\infty &&\mbox{since } \|T(\cdot)\|_2 \leq \sqrt{5} \|\cdot\|_\infty\\
&= \sqrt{10} \bigl(\|f\|_\infty |\nabla \phi| + \|\phi\|_\infty |\nabla f|\bigr) W_1(B,B').
\end{align*}
\end{proof}

It follows that persistence images are also stable.

\begin{theorem}\label{thm:image-stable-general}
The persistence image $I(\rho_B)$ is stable with respect to the 1-Wasserstein distance between diagrams. More precisely, if $A$ is the maximum area of any pixel in the image, $A'$ is the total area of the image, and $n$ is the number of pixels in the image, then
\begin{align*}
\| I(\rho_B) - I(\rho_{B'})\|_\infty &\leq \sqrt{10} A \bigl(\|f\|_\infty |\nabla \phi| + \|\phi\|_\infty |\nabla f|\bigr) W_1(B,B')\\
\| I(\rho_B) - I(\rho_{B'})\|_1 &\leq \sqrt{10} A' \bigl(\|f\|_\infty |\nabla \phi| + \|\phi\|_\infty |\nabla f|\bigr) W_1(B,B')\\
\| I(\rho_B) - I(\rho_{B'})\|_2 &\leq \sqrt{10n} A \bigl(\|f\|_\infty |\nabla \phi| + \|\phi\|_\infty |\nabla f|\bigr) W_1(B,B').
\end{align*}
\end{theorem}

The constant for the $L^2$ norm bound containing $\sqrt{n}$ goes to infinity as the resolution of the image increases. For this reason, in Theorem~\ref{thm:image-stable-Gaussians} we provide bounds with better constants in the specific case of Gaussian distributions.

\begin{proof}
Note for any pixel $p$ with area $A(p)$ we have 
\begin{align*}
|I(\rho_B)_p - I(\rho_{B'})_p| &= \Bigl| \iint_p \rho_B\; dydz- \iint_p \rho_{B'}\; dydx \Bigr|\\
&= \Bigl| \iint_p \rho_B-\rho_{B'}\; dydx \Bigr|\\
&\le A(p) \| \rho_B - \rho_{B'} \|_\infty\\
&\leq \sqrt{10} A(p) \bigl(\|f\|_\infty |\nabla \phi| + \|\phi\|_\infty |\nabla f|\bigr) W_1(B,B')&&\mbox{by Theorem~\ref{thm:surface-stable-general}.}
\end{align*}
Hence we have
\begin{align*}
\| I(\rho_B) - I(\rho_{B'})\|_\infty &\le \sqrt{10} A \bigl(\|f\|_\infty |\nabla \phi| + \|\phi\|_\infty |\nabla f|\bigr) W_1(B,B')\\
\| I(\rho_B) - I(\rho_{B'})\|_1 &\le \sqrt{10} A' \bigl(\|f\|_\infty |\nabla \phi| + \|\phi\|_\infty |\nabla f|\bigr) W_1(B,B')\\
\| I(\rho_B) - I(\rho_{B'})\|_2 &\le \sqrt{n} \| I(\rho_B) - I(\rho_{B'})\|_\infty\\
&\le \sqrt{10n} A \bigl(\|f\|_\infty |\nabla \phi| + \|\phi\|_\infty |\nabla f|\bigr) W_1(B,B').
\end{align*}
\end{proof}

\begin{rem}\label{rem:only-1-stable}
Recall $\cD$ is the set of all PDs. The kernel $k\colon\cD\times\cD\to\bR$ defined by $k(B,B')=\langle I(\rho_B), I(\rho_{B'})\rangle_{\bR^n}$ is non-trivial and additive, and hence Theorem~3 of \citet{reininghaus2015stable} implies that $k$ is not stable with respect to $W_p$ for any $1<p\le\infty$. That is, when $1<p\le\infty$ there is no constant $c$ such that for all $B, B'\in\cD$ we have $\|I(\rho_B)-I(\rho_{B'})\|_2 \le c W_p(B,B')$.
\end{rem}

\subsection{Stability for Gaussian distributions}\label{subsec:stability_Gaussian}

In this section, we provide stability results with better constants in the case of Gaussian distributions. With Gaussian distributions we can control not only the $L^\infty$ distance but also the $L^1$ distance between two persistence surfaces.

Our results for 2-dimensional Gaussians will rely on the following lemma for 1-dimensional Gaussians. 

\begin{lemma}\label{lem:l1surfbound_1D}
For $u,v\in\bR$, let $g_u,g_v \colon \bR \to \bR$ be the normalized 1-dimensional Gaussians, defined via $\displaystyle g_u(z) = \frac{1}{\sigma\sqrt{2\pi}}e^{-(z-u)^2/2 \sigma^2}$. If $a,b>0$, then \[\|ag_u - bg_v\|_1 \le |a-b|+\sqrt{\frac{2}{\pi}}\frac{\min\{a,b\}}{\sigma}|u-v|.\]
\end{lemma}

\begin{proof} Let $\text{Erf}(t) = \frac{2}{\sqrt{\pi}} \int_{0}^{t} e^{-u^2} du$. We show in Appendix~\ref{app:proofs} that \begin{equation}\label{eq:Erf}
\|ag_u - b g_v\|_1 = F(v-u),
\end{equation}
where $F\colon\bR\to\bR$ is defined by
\[ F(z) = \begin{cases}
|a-b| &\mbox{if }z=0\\
\left|a \text{Erf}\left(\frac{z^2+2\sigma^2\ln(a/b)}{z \sigma 2\sqrt{2}}\right)-b\text{Erf}\left(\frac{-z^2+2\sigma^2\ln(a/b)}{z \sigma 2\sqrt{2}}\right)\right| & \mbox{otherwise.}
\end{cases} \]
Furthermore, function $F$ is differentiable for $z\neq 0$. By searching for real roots of $F''$, one can show

\[ \| F' \|_\infty = \sqrt{\frac{2}{\pi}}\frac{\min\{a,b\}}{\sigma}, \quad\mbox{and hence } \quad F(z)\le|a-b|+\sqrt{\frac{2}{\pi}}\frac{\min\{a,b\}}{\sigma}|z|.\]
The result follows by letting $z=v-u$.
\end{proof}

\begin{lemma}\label{lem:l1surfbound_2D}
For $u,v\in\bR^2$, let $g_u,g_v \colon \bR^2 \to \bR$ be normalized 2-dimensional Gaussians. Then
\[ \displaystyle  \|f(u)g_u - f(v)g_v\|_1 \le \left(|\nabla f|+\sqrt{\frac{2}{\pi}}\frac{\min\{f(u),f(v)\}}{\sigma}\right)\|u-v\|_2. \]
\end{lemma}

The proof of Lemma \ref{lem:l1surfbound_2D} is shown in Appendix \ref{app:proofs} and uses a similar construction to that of Lemma \ref{lem:l1surfbound_1D}. We are prepared to prove the stability of persistence surfaces with Gaussian distributions.

\begin{theorem}\label{thm:surface-stable-Gaussians}
The persistence surface $\rho$ with Gaussian distributions is stable with respect to the 1-Wasserstein distance between diagrams: for $B,B'\in\cD$ we have
\[ \| \rho_B - \rho_{B'}\|_1 \leq \left(\sqrt{5}|\nabla f|+\sqrt{\frac{10}{\pi}}\frac{\|f\|_\infty}{\sigma}\right) W_1(B,B').\]
\end{theorem}

\begin{proof}
Since we assume $B$ and $B'$ consist of finitely many off-diagonal points, there exists a matching $\gamma$ that achieves the infimum in the Wasserstein distance. Then
\begin{align*}
&\| \rho_B - \rho_{B'} \|_1 = \left\|\sum_{u\in T(B)} f(u)g_u - \sum_{u\in T(B)} f(\gamma(u))g_{\gamma(u)}\right\|_1 \\
\leq& \sum_{u\in T(B)} \|f(u)g_u - f(\gamma(u))g_{\gamma(u)} \|_1 \\
\leq& \left(|\nabla f|+\sqrt{\frac{2}{\pi}}\frac{\|f\|_\infty}{\sigma}\right) \sum_{u\in T(B)} \|u-\gamma(u)\|_2 \quad\mbox{by Lemma~\ref{lem:l1surfbound_2D}, where }\min\{f(u),f(v)\}\le\|f\|_\infty \\
\leq& \left(\sqrt{5}|\nabla f|+\sqrt{\frac{10}{\pi}}\frac{\|f\|_\infty}{\sigma}\right)\sum_{u\in B} \|u-\gamma(u)\|_\infty \quad\mbox{since } \|T(\cdot)\|_2 \leq \sqrt{5} \|\cdot\|_\infty\\
=& \left(\sqrt{5}|\nabla f|+\sqrt{\frac{10}{\pi}}\frac{\|f\|_\infty}{\sigma}\right) W_1(B,B').
\end{align*}
\end{proof}

It follows that persistence images are also stable.

\begin{theorem}\label{thm:image-stable-Gaussians}
The persistence image $I(\rho_B)$ with Gaussian distributions is stable with respect to the 1-Wasserstein distance between diagrams.  More precisely,  
\begin{align*}
\| I(\rho_B) - I(\rho_{B'})\|_1 &\leq \left(\sqrt{5}|\nabla f|+\sqrt{\frac{10}{\pi}}\frac{\|f\|_\infty}{\sigma}\right) W_1(B,B')\\
\| I(\rho_B) - I(\rho_{B'})\|_2 &\leq \left(\sqrt{5}|\nabla f|+\sqrt{\frac{10}{\pi}}\frac{\|f\|_\infty}{\sigma}\right) W_1(B,B')\\
\| I(\rho_B) - I(\rho_{B'})\|_\infty &\leq \left(\sqrt{5}|\nabla f|+\sqrt{\frac{10}{\pi}}\frac{\|f\|_\infty}{\sigma}\right) W_1(B,B').
\end{align*}
\end{theorem}

\begin{proof}
We have
\begin{align*}
\| I(\rho_B) - I(\rho_{B'})\|_1 &= \sum_p \Bigl| \iint_p \rho_B\; dydz- \iint_p \rho_{B'}\; dydx \Bigr| \le \iint_{\bR^2} |\rho_B-\rho_{B'}|\; dydz\\
&=\|\rho_B-\rho_{B'}\|_1 \leq \left(\sqrt{5}|\nabla f|+\sqrt{\frac{10}{\pi}}\frac{\|f\|_\infty}{\sigma}\right) W_1(B,B')
\end{align*}
by Theorem~\ref{thm:surface-stable-Gaussians}. The claim follows since $\| \cdot \|_2 \le \| \cdot \|_1$ and $\| \cdot \|_\infty \le \| \cdot \|_1$ for vectors in $\bR^n$.
\end{proof}

\section{Experiments}\label{sec:experiments}

In order to assess the addded value of our vector representation of PDs, we compare the performance of PDs, PLs, and PIs in \S\ref{sec:kmedoids} in a classification task for a synthetic data set consisting of point clouds sampled from six different topological spaces using $K$-medoids, which utilizes only the internal dissimilarities of each of the topological data descriptions to classify. We find that PIs produce consistently high classification accuracy and, furthermore, the computation time for PIs is significantly faster than computing bottleneck or Wasserstein distances between PDs. In \S \ref{subsec:param_selection}, we explore the impact that the choices of parameters determining our PIs have on classification accuracy.  We find that the accuracy is insensitive to the particular choices of PI resolution and distribution variance. In \S\ref{feature_selection:sec}, we combine PIs with a sparse support vector machine classifier to identify the most strongly differentiating pixels for classification; this is an example of a ML task which is facilitated by the fact that PIs are finite vectors. Finally, as a novel machine learning application, we illustrate the utility of PIs to infer dynamical parameter values in both continuous and discrete dynamical systems: a discrete time system called the linked twist map in \S \ref{sec:linkedtwist}, and a partial differential equation called the anisotropic Kuramoto-Sivashinsky equation in \S \ref{sec:aks}. 

\subsection{Comparison of PDs, PLs, and PIs using $K$-medoids Classification}
\label{sec:kmedoids}

Our synthetic dataset consists of six shape classes: a solid cube, a circle, a sphere, three clusters, three clusters within three clusters, and a torus. Given a level of Gaussian random noise, we produce 25 point clouds of 500 randomly sampled noisy points from each of the six shapes; giving 150 point clouds in total. We then compute the $H_0$ and $H_1$ PDs for the Vietoris--Rips filtration (\S\ref{PD_Data}) built from each point cloud which have been endowed with the ambient Euclidean metric on $\bR^3.$

Our goal is to compare various methods for transforming PDs into distance matrices to be used to establish proximity of topological features extracted from data.  We create $3^2 \cdot 2^2=36$ distance matrices of size $150\times150$, using three choices of representation (PDs, PLs, PIs), three choices of metric ($L^1$, $L^2$, $L^\infty$)\footnote{The $L^1$, $L^2$, $L^\infty$ distances on PDs are more commonly known as the 1-Wasserstein, 2-Wasserstein, and bottleneck distances.}, two choices of Gaussian noise ($\eta=0.05$, $0.1$), and two homological dimensions ($H_0$, $H_1$). For example, the PD, $H_1$, $L^2$, $\eta=0.1$,  distance matrix contains the 2-Wasserstein distances between the $H_1$ PDs for the random point clouds with noise level $0.1$. By contrast, the PI, $H_1$, $L^2$, $\eta=0.1$ distance matrix contains all pairwise $L^2$ distances between the PIs\footnote{For PIs in this experiment, we use variance $\sigma=0.1$, resolution $20\times20$, and the weighting function defined in \S\ref{sec:persistenceimages}.} produced from the $H_1$ PDs with noise level $0.1$.

We first compare these distance matrices based on how well they classify the random point clouds into shape classes via $K$-medoids clustering \citep{kaufman1987clustering,park2009simple}. $K$-medoids produces a partition of a metric space into $K$ clusters by choosing $K$ points from the dataset called \emph{medoids} and assigning each metric space point to its closest medoid. The \emph{score} of such a clustering is the sum of the distances from each point to its closest medoid. The desired output of $K$-medoids is the clustering with the minimal clustering score. Unfortunately, an exhaustive search for the global minimum is often prohibitively expensive. A typical approach to search for this global minimum is to choose a large selection of $K$ random initial medoids, improve each selection of medoids iteratively in rounds until the clustering score stabilizes and then return the identified final clustering with the lowest score for each initialization. In our experiments, we choose 1,000 random initial selections of $K=6$ medoids (as there are six shape classes) for each distance matrix, improve each selection of medoids using the Voronoi iteration method \citep{park2009simple}, and return the clustering with the lowest classification score. To each $K$-medoids clustering we assign an accuracy which is equal to the percentage of random point clouds identifed with a medoid of the same shape class. In Table~\ref{table:Kmedoids}, we report the classification accuracy of the $K$-medoids clustering with the lowest clustering score, for each distance matrix. 

Our second criterion for comparing methods to produce distance matrices is computational efficiency. In Table~\ref{table:Kmedoids}, we report the time required to produce each distance matrix, starting with 150 precomputed PDs as input. In the case of PLs and PIs, this time includes the intermediate step of transforming each PD into the alternate representation, as well as computing the pairwise distance matrix. All timings are computed on a laptop with a 1.3 GHz Intel Core i5 processor and 4 GB of memory. We compute bottleneck, 1-Wasserstein, and 2-Wasserstein distance matrices using the software of \citet{kerbergeometry}. For PL computations, we use the Persistence Landscapes Toolbox by \citet{bubenik2014persistence}. Our MATLAB code for producing PIs is publically available at \url{https://github.com/CSU-TDA/PersistenceImages}.

\begin{table}[h]
\caption{Comparing classification accuracy and times of PDs, PLs, and PIs. The timings contain the computation time in seconds for producing a $150\times150$ distance matrix from 150 precomputed PDs. In the case of PLs and PIs, this requires first transforming each PD into its alternate representation and then computing a distance matrix. We consider 36 distinct distance matrices: three representations (PDs, PLs, PIs), two homological dimensions ($H_0$, $H_1$), three choices of metric ($L^1$, $L^2$, $L^\infty$), and two levels of Gaussian noise ($\eta=0.05$, $0.1$).}
\label{table:Kmedoids}
\renewcommand{\arraystretch}{1.3}
\captionsetup{justification=centering}
\centering
\scalebox{0.95}{\centering
\begin{tabular}{|l|c|c|c|c|}   
\hline   
\bfseries Distance Matrix& \shortstack{Accuracy \\ (Noise 0.05)} & \shortstack{Time \\ (Noise 0.05)} & \shortstack{Accuracy \\ (Noise 0.1)} & \shortstack{Time \\ (Noise 0.1)} \\ \hline \hline
PD, $H_0$, $L^1$ & 96.0\% & 37346s & 96.0\% & 42613s \\ \hline
PD, $H_0$, $L^2$ & 91.3\% & 24656s & 91.3\% & 25138s \\ \hline
PD, $H_0$, $L^{\infty}$ & 60.7\% & 1133s & 63.3\% & 1149s \\ \hline
PD, $H_1$, $L^1$ & 100\% & 657s & 96.0\% & 703s \\ \hline
PD, $H_1$, $L^2$ & 100\% & 984s & 97.3\% & 1042s \\ \hline
PD, $H_1$, $L^{\infty}$ & 81.3\% & 527s & 66.7\% & 564s \\ \hline \hline
PL, $H_0$, $L^1$ & 92.7\% & 29s & 96.7\% & 33s \\ \hline
PL, $H_0$, $L^2$ & 77.3\% & 29s & 82.0\% & 34s \\ \hline
PL, $H_0$, $L^\infty$ & 60.7\% & 2s & 63.3\% & 2s \\ \hline
PL, $H_1$, $L^1$ & 83.3\% & 36s & 80.7\% & 48s \\ \hline
PL, $H_1$, $L^2$ & 83.3\% & 50s & 66.7\% & 69s \\ \hline
PL, $H_1$, $L^\infty$ & 74.7\% & 8s & 66.7\% & 9s \\ \hline \hline
PI, $H_0$, $L^1$ & 93.3\% & 9s & 95.3\% & 9s \\ \hline
PI, $H_0$, $L^2$ & 92.7\% & 9s & 95.3\% & 9s \\ \hline
PI, $H_0$, $L^\infty$ & 94.0\% & 9s & 96.0\% & $9$s \\ \hline
PI, $H_1$, $L^1$ & 100\% & 17s & 95.3\% & 18s \\ \hline
PI, $H_1$, $L^2$ & 100\% & 17s & 96.0\% & 18s \\ \hline
PI, $H_1$, $L^\infty$ & 100\% & 17s & 96.0\% & 18s \\ \hline
\end{tabular}
}
\label{bigtable}
\end{table}

We see in Table~\ref{table:Kmedoids} that PI distance matrices have higher classification accuracy than nearly every PL distance matrix, and higher classification accuracy than PDs in half of the trials. 
Furthermore, the computation times for PI distance matrices are significantly lower than the time required to produce distance matrices from PDs using the bottleneck or $p$-Wasserstein metrics. In this experiment, persistent images provide a representation of persistent diagrams which is both useful for the classification task and also computationally efficient. 


\subsection{Effect of PI Parameter Choice}
\label{subsec:param_selection}

In any system that relies on multiple parameters, it is important to understand the effect of parameter values on the system. As such, we complete a search of the parameter space used to generate PIs on the shape dataset described in \S\ref{sec:kmedoids} and measure $K$-medoids classification accuracy as a function of the parameters. We explore 20 different resolutions (ranging from $5 \times 5$ to $ 100 \times 100$ in increments of 5), use a Gaussian function with 20 different choices of variance (ranging from 0.01 to 0.2 in increments of 0.01), and the weighting function described in \S\ref{sec:persistenceimages}.
\begin{figure}[h]
\captionsetup[subfigure]{justification=centering}
\centering
\includegraphics[width=.75\textwidth]{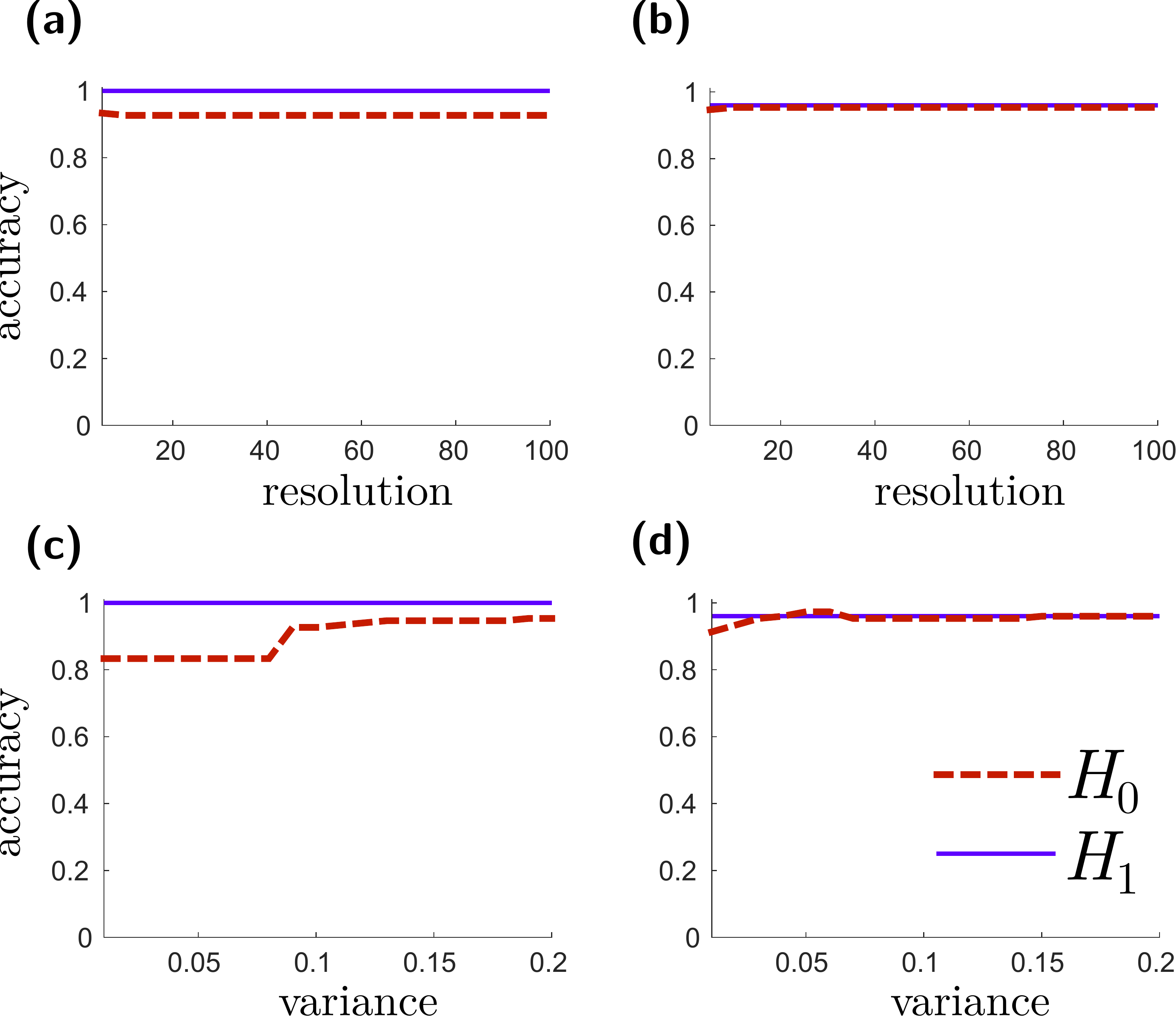}
\caption{$K$-medoids classification accuracy as a function of resolution and variance for the dataset of six shape classes. First column: noise level $\eta = 0.05$. Second column: noise level $\eta = 0.1$. First row: fixed variance 0.1 with resolutions ranging from $5\times 5$ to $100\times 100$ in increments of 5. Second row: fixed resolution $20\times 20$ with variances ranging from 0.01 to 0.2 in increments of 0.01. 
}
\label{fig:Clustering_kmed}
\end{figure}
For each set of parameters, we compute the classification accuracy of the $K$-medoids clustering with the minimum clustering score on the two sets of noise levels for the homology dimensions $H_0$ and $H_1$. We observe that the classification accuracy is insensitive to the choice of resolution and variance. 
The plots in Figure~\ref{fig:Clustering_kmed} are characteristic of the 2-dimensional accuracy surface over all combinations of parameters in the ranges of variances and resolutions we tested. In an application to archeology, \citet{zeppelzauer2016topological} find a similar robustness of PIs to the choices of resolution and variance.

\subsection{Differentiating Homological Features by Sparse Support Vector Machine}
\label{feature_selection:sec}

The 1-norm regularized \textit{linear} support vector machine (SVM), a.k.a.\ sparse SVM (SSVM) classifies data by generating a separating hyperplane that depends on very few input space features \citep{Bradley1998,Zhu2003,Zhang2010b}. Such a model can be used for reducing data dimension or selecting discriminatory features. Note that linear SSVM feature selection is implemented on vectors and, therefore, can be used on our PIs to select discriminatory pixels during classification. 
Other PD representations in the literature \citep{reininghaus2015stable,pachauri2011topology} are designed to use kernel ML methods, such as \textit{kernel} (nonlinear) SVMs. However, constructing kernel SVM classifiers using the 1-norm results in minimizing the number of kernel functions, not the number of features in the input space (i.e. pixels in our application) \citep{Fung2004}. Hence, for the purpose of feature selection or, more precisely, PI pixel selection, we employ the linear SSVM.

We adopt the one-against-all (OAA) SSVM on the sets of $H_0$ and $H_1$ PIs from the six class shape data. In a one-against-all SSVM, there is one binary SSVM for each class to separate members of that class from members of all other classes. The PIs were generated using resolution $20 \times 20$, variance $0.0001$, and noise level $0.05$. Note that because of the resolution parameter choice of $20 \times 20$, each PI is a $400$-dimensional vector, and the selected features will be a subset of indices corresponding to pixels within the PI. Using 5-fold cross-validated SSVM resulted in $100\%$ accuracy comparing six sparse models with indications of the discriminatory features. Feature selection is performed by retaining the features (again, in this application, pixels) with non-zero SSVM weights, determined by magnitude comparison using weight ratios; for details see \citet{chep}. Figure~\ref{fig:FeatureSelection1} provides two examples, indicating the pixels of $H_1$ PIs that discriminate circles and tori from the other classes in the synthetic data set.

\begin{figure}[H]
\centering
\includegraphics[width=0.6\textwidth]{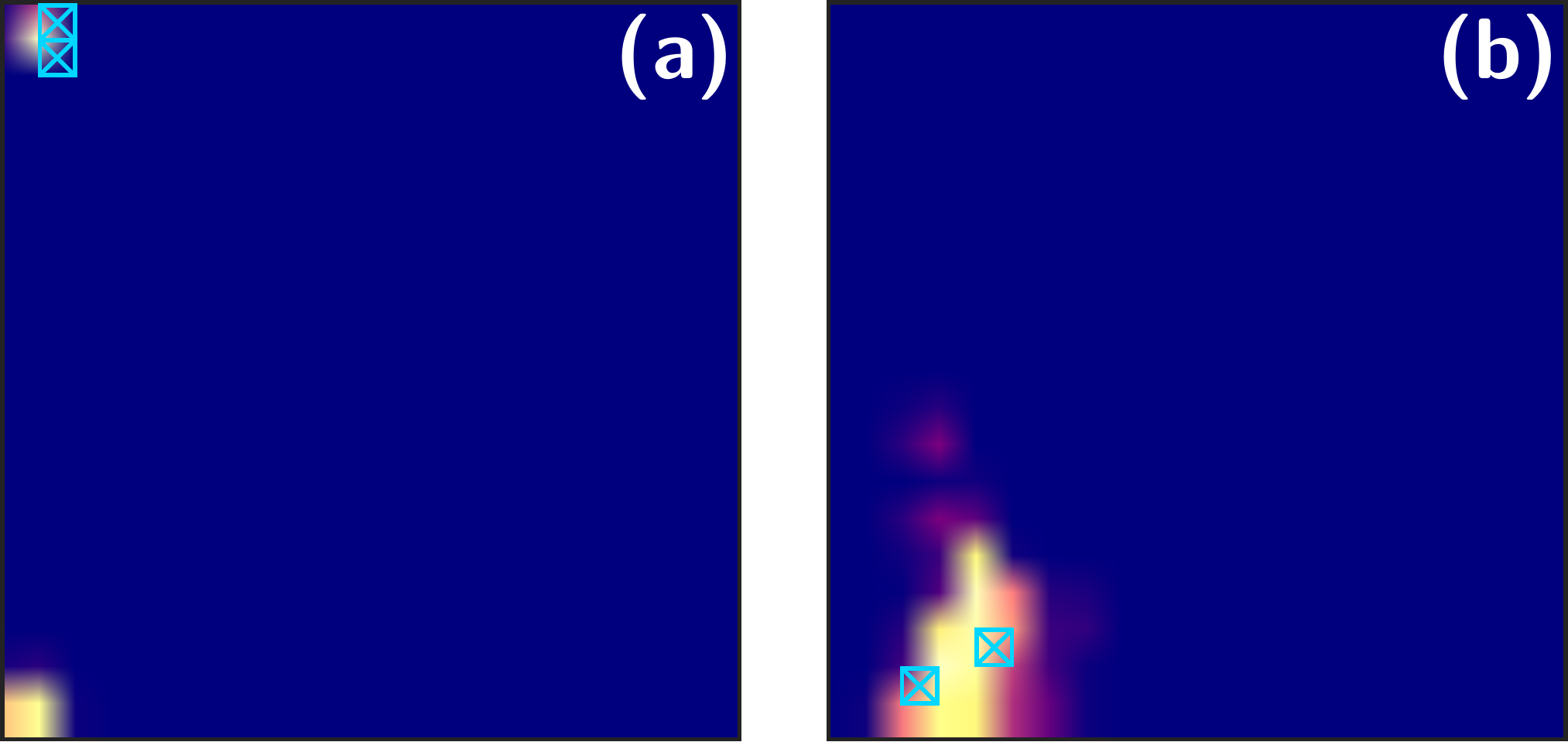}
\caption{SSVM-based feature (pixel) selection for $H_1$ PIs from two classes of the synthetic data. Selected pixels are marked by blue crosses. (a) A noisy circle with the two selected pixels (indices 21 and 22 out of 400). (b) A noisy torus with the two selected pixels (indices 59 and 98 out of 400). The PI parameters used are resolution $20 \times 20$ and variance $10^{-4}$, for noise level $0.05$.}
\label{fig:FeatureSelection1}
\end{figure}

Feature selection produces highly interpretable results. The discriminatory pixels in the $H_1$ PIs that separate circles from the other classes correspond to the region where highly persistent $H_1$ topological features exist across all samples of a noisy circle (highlighted in Figure~\ref{fig:FeatureSelection1}a). Alternatively, the discriminatory pixels in $H_1$ PIs that separate tori from the other classes correspond to points of short to moderate persistence (see Figure~\ref{fig:FeatureSelection1}b). In this way, Figure~\ref{fig:FeatureSelection1}b reiterates an observation of \citet{bendich2014persistent} that points of short to moderate persistence can contain important discriminatory information. Similar conclusions can be drawn from the discriminatory pixels of others classes (Appendix~\ref{app:featureselection}). Our classification accuracy of 100\% is obtained using only those pixels selected by SSVM (a cumulative set of only 10 distinct pixels).

\subsection{Application: Determination of Dynamical System Parameters}

Models of dynamic physical phenomenon rarely agree perfectly with the reality they represent.  This is often due to the presence of poorly-resolved (or poorly-understood) processes which are parameterized rather than treated explicitly.  As such, determination of the influence of a model parameter --- which may itself be an incompletely-described conglomeration of several physical parameters --- on model dynamics is a mainstay of dynamical system analysis.  In the case of fitting a dynamic model to data, i.e. explicit determination of optimal model parameters, a variety of techniques exist for searching through parameter space, which often necessitate costly simulations.  Furthermore, such approaches struggle when applied to models exhibiting sensitivity to initial conditions.  We recast this problem as a machine-learning exercise based on the hypotheses that model parameters will be reflected directly in dynamic data in a way made accessible by persistent homology.

\subsubsection{A discrete dynamical model} \label{sec:linkedtwist}
We approach a classification problem with data arising from the linked twist map, a discrete dynamical system modeling fluid flow. \citet{DNAHertzsch} use the linked twist map to model flows in DNA microarrays with a particular interest in understanding turbulent mixing. This demonstrates a primary mechanism giving rise to chaotic advection. The linked twist map is a Poincar\'e section of \textit{eggbeater-type flow} \citep{DNAHertzsch} in continuous dynamical systems.  The Poincar\'e section captures the behavior of the flow by viewing a particle's location at discrete time intervals.  The linked twist map is given by the discrete dynamical system
\begin{align*}
x_{n+1}&=x_n+r y_n(1-y_n) \mod 1 \\
y_{n+1}&=y_n+r x_n(1-x_n)  \mod 1,
\end{align*}
where $r$ is a positive parameter.  For some values of $r$, the orbits $\{ (x_n, y_n) \; : \; n = 0,
\ldots, \infty \}$ are dense in the domain.  However, for other parameter values, voids form.  In either case, the truncated orbits $\{ (x_n, y_n) \; : \; n = 0, \ldots, N \in \mathbb{N} \}$  exhibit complex structure. 

For this experiment, we choose a set of parameter values, $r=$ 2.5, 3.5, 4.0, 4.1 and 4.3, which produce a variety of orbit patterns. For each parameter value, 50 randomly-chosen initial conditions are selected, and 1000 iterations of the linked twist map are used to generate point clouds in $\mathbb{R}^2$.  Figure \ref{fig:LTM} shows examples of typical orbits generated for each parameter value.  The goal is to classify the trials by parameter value using PIs to capitalize on distinguishing topological features of the data. We use resolution $20\times 20 $ and a Gaussian with variance $\sigma=0.005$ to generate the PIs. These parameters were chosen after a preliminary parameter search and classification effort. Similar results hold for a range of PI parameter values. 

\begin{figure} [h]
\captionsetup[subfigure]{justification=centering}
\centering
\includegraphics[width=\textwidth]{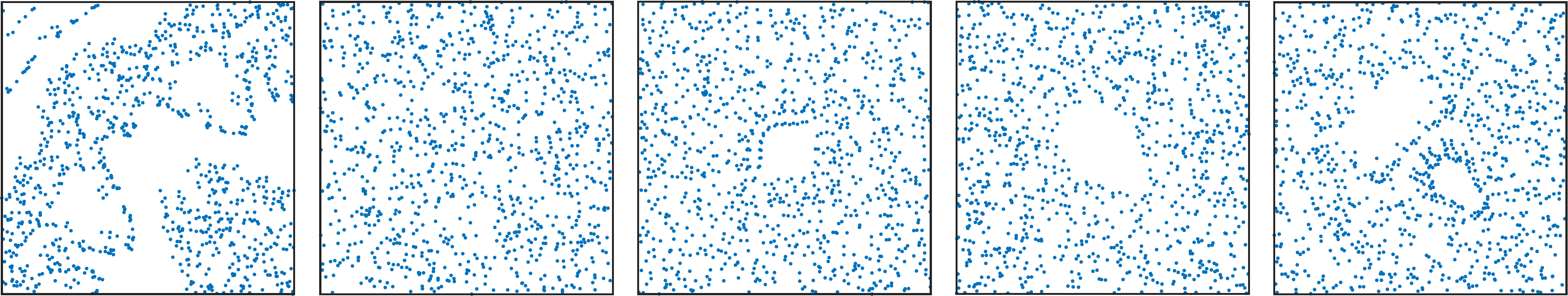}
\caption{Examples of the first 1000 iterations, $\{ (x_n, y_n) \; : \; n = 0, \ldots, 1000 \}$,  of the linked twist map with parameter values $r=$ 2, 3.5, 4.0, 4.1 and 4.3, respectively.}
\label{fig:LTM}
\end{figure}

\begin{figure}[h!]
\centering
\includegraphics[width=\textwidth]{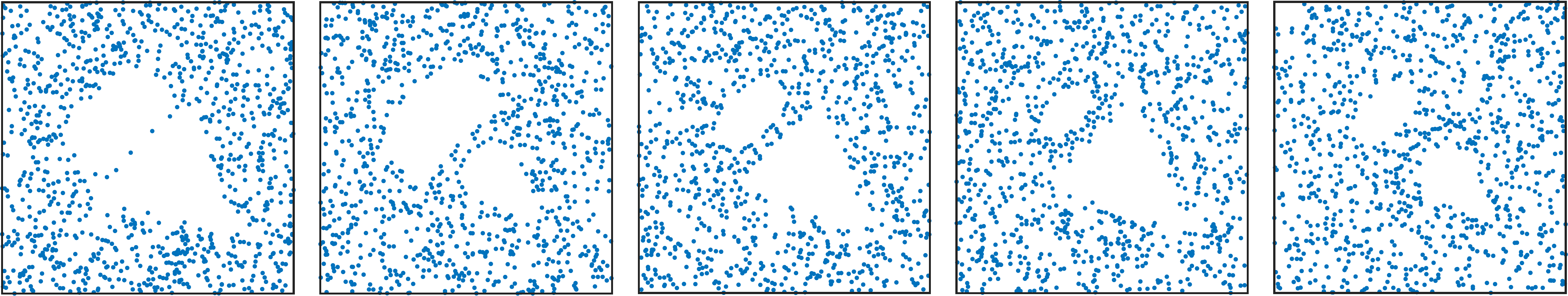}
\caption{Truncated orbits, $\{ (x_n, y_n) \; : \; n = 0, \ldots, 1000 \}$, of the linked twist map with fixed $r=4.3$ for different initial conditions $(x_0,y_0)$.}
\label{fig:LTM4-3var}
\end{figure}
For a fixed $r$ parameter value and a large number of points (many thousands), the patterns in the distributions of iterates show only small visible variations for different choices of the initial condition $(x_0,y_0)$. However, with few points, such as in Figure \ref{fig:LTM4-3var}, there are more significant variations in the patterns for different choices of initial conditions, making classification more difficult. 


We perform classification and cross-validation with a discriminant subspace ensemble.  This ML algorithm trains many ``weak'' learners on randomly chosen subspaces of the data (of a fixed dimension), and classifies and  assigns a score to each point based on the current subspace. The final classification arises from an average of the scores of each data point over all learners \citep{ho1998random}. We perform 10 trials and average the classification accuracies.  For the concatenated $H_0$ and $H_1$ PIs, this method achieves a classification accuracy of 82.5\%; compared to 49.8\% when using only $H_0$ PIs, and 65.7\% using $H_1$ PIs.  This experiment highlights two strengths of PIs: they offer flexibility in choosing a ML algorithm that is well suited to the data under consideration, and homological information from multiple dimensions may be leveraged simultaneously for greater discriminatory power.
 

This application is a brief example of the utility of PIs in classification of data from dynamical systems and modeling real-world phenomena, which provides a promising direction for further applications of PIs.

\subsubsection{A partial differential equation} \label{sec:aks}

The Kuramoto-Sivashinsky (KS) equation is a partial differential equation for a function $u(x,y,t)$ of spatial variables $x,y$ and time $t$ that has been independently derived in a variety of problems involving pattern formation in extended systems driven far from equilibrium.  Applications involving surface dynamics include surface nanopatterning by ion-beam erosion \citep{cb95,Motta12}, epitaxial growth \citep{villain91,wolf91,rost95}, and solidification from a melt \citep{golovin98}.  In these applications, the nonlinear term in the KS equation may be anisotropic, resulting in the anisotropic Kuramoto-Sivashinsky (aKS) equation
\begin{equation}
\label{aKS}
\frac{\partial}{\partial t} u = - \nabla^2 u - \nabla^2 \nabla^2 u + r \left(\frac{\partial }{\partial x} u \right)^2 + \left( \frac{\partial}{\partial y} u \right)^2,
\end{equation}
where $\nabla^2 = \frac{\partial^2}{ \partial x^2}  + \frac{\partial^2}{\partial y^2} $, and the real parameter $r$ controls the degree of anisotropy.  At a fixed time $t^*$, $u(x,y,t^*)$ is a patterned surface (periodic in both $x$ and $y$) defined over the $(x,y)$-plane. Visibly, the anisotropy appears as a slight tendency for the pattern to be elongated in the vertical or horizontal direction.      

\begin{figure}[h!]
\centering
\includegraphics[width=.9\textwidth]{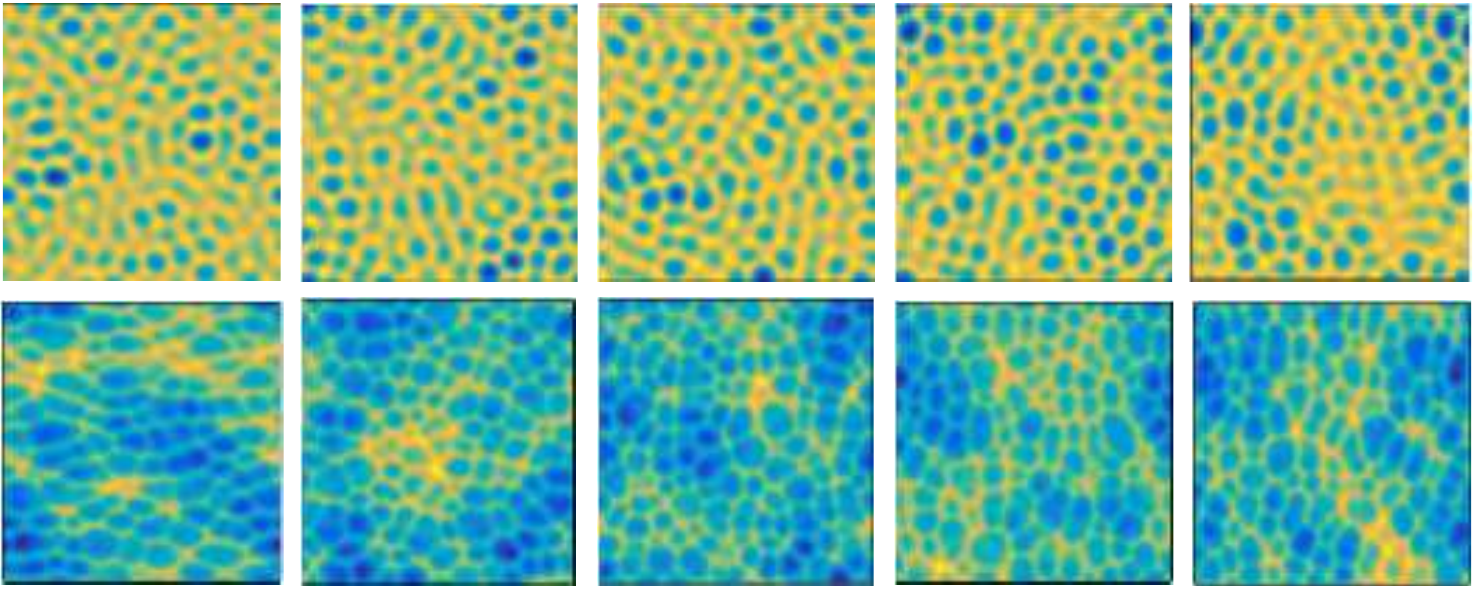}
\caption{Plots of height-variance-normalized surfaces $u(x,y,\cdot)$ resulting from numerical simulations of the aKS equation \eqref{aKS}.  Each column represents a different parameter value: (from left) $r=1$, 1.25, 1.5, 1.75 and 2. Each row represents a different time: $t=3$ (top) and $t=5$ (bottom). By $t=5$ any anisotropic elongation of the surface pattern has visibly stabilized.}
\label{fig:aks-ex}
\end{figure}

Numerical simulations of the aKS equation for a range of parameter values (columns) and simulation times (rows) are shown in Figure~\ref{fig:aks-ex}.  For all simulations, the initial conditions were low-amplitude white noise.  We employed a Fourier spectral method with periodic boundary conditions on a $512 \times 512$ spatial grid, with a fourth-order exponential time differencing Runge-Kutta method for the time stepping. Five values for the parameter $r$ were chosen, namely $r=1$, 1.25, 1.5, 1.75 and 2, and thirty trials were performed for each parameter value. Figure~\ref{fig:quiz} shows the similarity between surfaces associated to two parameter values $r=1.75$ and $r=2$ at an early time.

\begin{figure}[h!]
\captionsetup{singlelinecheck=off}
\centering
\includegraphics[width=.9\textwidth]{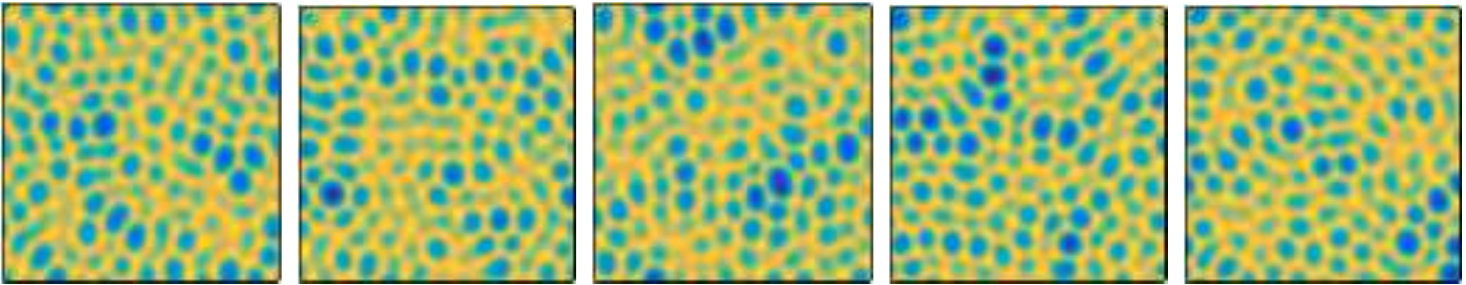}
\caption[NecessaryName-DontDelete]{To illustrate the difficulty of our classification task, consider five instances of surfaces $u(x,y,3)$ for $r=1.75$ or $r=2$, plotted on the same color axis. These surfaces are found by numerical integration of Equation~\eqref{aKS}, starting from random initial conditions. Can you group the images by eye?\\
\rotatebox{180}{\textbf{Answer:} (from left)  $r=1.75, 2, 1.75, 2, 2$.}}
\label{fig:quiz}
\end{figure}

We aim to identify the anisotropy parameter for each simulation using snapshots of surfaces $u(x,y,\cdot)$ as they evolve in time.  Inference of the parameter using the surface alone proves difficult for several reasons.  First, Equation~\eqref{aKS} exhibits sensitivity to initial conditions: initially nearby solutions diverge quickly.  Second, although the surface $u(x,y,t^*)$ at a fixed time is an approximation due to the finite discretization of its domain, the spatial resolution is still very large: in fact, these surfaces may be thought of as points in $\mathbb{R}^{266144}$.  We were unable to perform standard classification techniques in this space. It was therefore necessary to perform some sort of dimension reduction.  One such method is to simply `resize' the surface by coarsening the discretization of the spatial domain after computing the simulation at a high resolution by replacing a block of grid elements with their average surface height.  The surfaces were resized in this way to a resolution of $10 \times 10$ and a subspace discriminant ensemble was used to perform classification.  Unsurprisingly, this method performs very poorly at all times (first row of Table~\ref{aksresults}).

The anisotropy parameter also influences the mean and amplitude of the surface pattern.  We eliminate differences in the mean by mean-centering each surface after the simulation. To assess the impact of the variance of surface height on our task, classification was performed using a normal distribution-based classifier built on the variances of the surface heights. In this classifier, a normal distribution was fit to a training set of 2/3 of the variances for each parameter value, and the testing data was classified based on a $z$-test for each of the different models. That is, a $p$-value for each new variance was computed for membership to the five normal distributions (corresponding to the five parameter choices of $r$), and the surface was classified based on the model yielding the highest $p$-value. After the pattern has more fully emerged (by, say, time $t=5$) this method of classification yields 75$\% $ accuracy\footnote{Accuracy reported is averaged over 100 different training and testing partitions.}, as shown in Table~\ref{aksresults}. However, early on in the formation of the pattern, this classifier performs very poorly because height variance is not yet a discriminating feature. Figure \ref{fig:variance_dist} shows the normal distribution fit to the variance of the surfaces for each parameter value at times $t=3$ and 5, and illustrates why the variance of surface height is informative only after a surface is allowed to evolve for a sufficiently long time. 

\begin{figure}[h!]
\begin{centering}
\includegraphics[width=0.9\textwidth]{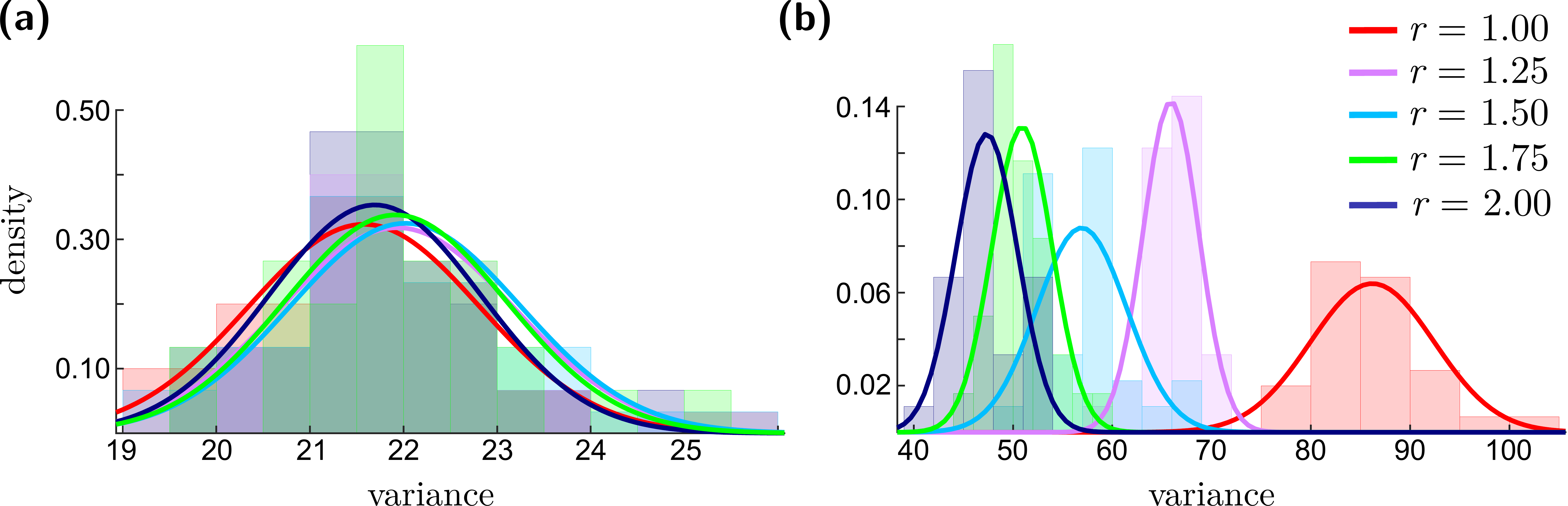}
\caption{Histograms of the variances of surface heights for each parameter value, and the normal distribution fit to each histogram, for times (a) $t=3$ and (b) $t=5$.}
\label{fig:variance_dist}
\end{centering}
\end{figure}

Variance of a surface is reflected in its sublevel set filtration (see \S\ref{app:PD_Functions} for more details) PD.  Yet, the PD and the subsequent PI contain additional topological structure, which may reveal other influences of the anisotropy parameter on the evolution of the surface.  Persistence diagrams were computed using the sublevel set filtration, and PIs were generated with resolution  $10 \times 10$ and a Gaussian with variance $\sigma=0.01$.  We think of our pipeline to a PI as a dimensionality reduction in this case, taking a surface which in actuality is a very high-dimensional point and producing a much lower dimensional one that retains meaningful characteristics of the original surface.  

We again use a subspace discriminant ensemble to classify PIs by parameter.  Table \ref{aksresults} compares these results to the same technique applied to low dimensional approximations of the raw surfaces and the normal distribution-based classifier built from surface variance alone.  At each time in the system evolution, the best classification accuracy results from using PIs, improving accuracies over using either low resolution approximations of the surfaces or variance of surface height alone by at least 20$\%$, including at early times in the evolution of the surface when pattern amplitudes are not visibly differentiated (see Figure~\ref{fig:quiz}).  We postulate that PIs capture more subtle topological information that is useful for identifying the parameter used to generate each surface. 

As we observed in \S\ref{sec:linkedtwist}, concatenating $H_0$ and $H_1$ PIs can notably improve the classification accuracy over either feature vector individually.  We again note that classification accuracy appears insensitive to the PI parameters.  For  example, when the variance of the Guassians used to generate the PIs was varied from 0.0001 to 0.1, the classification accuracy of the H$_0$ PIs, changed by less than one percentage point. The classification accuracy for H$_1$ fluctuated in a range of approximately five points. For a fixed variance, when the resolution of the image was varied from 5 to 20, the H$_0$ accuracy varied by little more than three points until the accuracy dropped by six points for a resolution of 25. 

PIs performed remarkably well in this classification task, allowing one to capitalize on subtle structural differences in the patterns and significantly reduce the dimension of the data for classification. There is more to be explored in the realm of pattern formation and persistence that is outside the scope of this paper.

\begin{table}[t!]
\centering
\caption{Classification accuracies at different times of the aKS solution, using different classification approaches. Classification of times $t=$ 15 and 20 result in accuracies similar to $t=$10.}
\begin{tabular}{|c|c|c|c|}   
\hline   
\bfseries{Classification Approach}& \shortstack{Time \\t=3} & \shortstack{Time \\t=5} & \shortstack{Time \\t=10} \\
\hline \hline
Subspace Discriminant Ensemble, Resized Surfaces & 26.0 \% & 19.3\% & 19.3 \% \\
\hline
 Variance Normal Distribution Classifier & 20.74\% & 75.2\%  & 77.62 \% \\
\hline
\hline
Subspace Discriminant Ensemble, $H_0$ PIs & 58.3 \% & 96.0 \% & 94.7 \% \\
\hline
Subspace Discriminant Ensemble, $H_1$ PIs & 67.7 \% & 87.3 \% & 93.3\% \\
\hline
Subspace Discriminant Ensemble, $H_0$ and $H_1$ PIs & 72.7 \% & 95.3 \%  &  97.3 \% \\
\hline

\end{tabular}
\label{aksresults}
\end{table}


\section{Conclusion}
PIs offer a stable representation of the topological characteristics captured by a PD. Through this vectorization, we open the door to a myriad of ML tools. This serves as a vital bridge between the fields of ML and topological data analysis and enables one to capitalize on topological structure (even in multiple homological dimensions) in the classification of data. 

We have shown PIs yield improved classification accuracy over PLs and PDs on sampled data of common topological spaces at multiple noise levels using $K$-medoids. Additionally, computing distances between PIs requires significantly less computation time compared to computing distances between PDs, and comparable computation times with PLs. Through PIs, we have gained access to a wide variety of ML tools, such as SSVM which can be used for feature selection. Features (pixels) selected as discriminatory in a PI are interpretable because they correspond to regions of a PD. We have explored datasets derived from dynamical systems and illustrated that topological information of solutions can be used for inference of parameters since PIs encapsulate this information in a form amenable to ML tools, resulting in high accuracy rates for data that is difficult to classify. 

The classification accuracy is robust to the choice of parameters for building PIs, providing evidence that it is not necessary to perform large-scale parameter searches to achieve reasonable classification accuracy. This indicates the utility of PIs even when there is not prior knowledge of the underlying data (i.e.\ high noise level, expected holes, etc.). The flexibility of PIs allows for customization tailored to a wide variety of real-world data sets.

\vskip .5cm
\noindent \textbf{Acknowledgments:} We would like to acknowledge the research group of Paul Bendich at Duke University for allowing us access to a persistent homology package which greatly reduced computational time and made analysis of large point clouds feasible. This code can be accessed via GitLab after submitting a request to Paul Bendich. This research is partially supported by the National Science Foundation under Grants No. DMS-1228308, DMS-1322508, NSF DMS-1115668, NSF DMS-1412674, and DMR-1305449 as well as the DOD-USAF under Award Number FA9550-12-1-0408. 

\bibliographystyle{plainnat}
\bibliography{bibliography}

\begin{thebibliography}{53}
\providecommand{\natexlab}[1]{#1}
\providecommand{\url}[1]{\texttt{#1}}
\expandafter\ifx\csname urlstyle\endcsname\relax
  \providecommand{\doi}[1]{doi: #1}\else
  \providecommand{\doi}{doi: \begingroup \urlstyle{rm}\Url}\fi

\bibitem[Adcock et~al.(2013)Adcock, Carlsson, and Carlsson]{adcock2012ring}
Aaron Adcock, Erik Carlsson, and Gunnar Carlsson.
\newblock The ring of algebraic functions on persistence bar codes.
\newblock \emph{arXiv preprint arXiv:1304.0530}, 2013.

\bibitem[Bendich(2009)]{bendichThesis}
Paul Bendich.
\newblock \emph{Analyzing Stratified Spaces Using Persistent Versions of
  Intersection and Local Homology}.
\newblock PhD thesis, Duke University, 2009.

\bibitem[Bendich et~al.(2014)Bendich, Chin, Clarke, deSena, Harer, Munch,
  Newman, Porter, Rouse, Strawn, and Watkins]{bendich2014topologicalLiz}
Paul Bendich, Sang Chin, Jesse Clarke, Jonathan deSena, John Harer, Elizabeth
  Munch, Andrew Newman, David Porter, David Rouse, Nate Strawn, and Adam
  Watkins.
\newblock Topological and statistical behavior classifiers for tracking
  applications.
\newblock \emph{arXiv preprint arXiv:1406.0214}, 2014.

\bibitem[Bendich et~al.(2015)Bendich, Marron, Miller, Pieloch, and
  Skwerer]{bendich2014persistent}
Paul Bendich, JS~Marron, Ezra Miller, Alex Pieloch, and Sean Skwerer.
\newblock Persistent homology analysis of brain artery trees.
\newblock \emph{Annals of Applied Statistics}, 2015.
\newblock To appear.

\bibitem[Bradley and Mangasarian(1998)]{Bradley1998}
Paul~S Bradley and Olvi~L Mangasarian.
\newblock Feature selection via concave minimization and support vector
  machines.
\newblock In \emph{Machine Learning Proceedings of the Fifteenth International
  Conference}, ICML 1998, pages 82--90, 1998.

\bibitem[Bubenik(2015)]{bubenik2015statistical}
Peter Bubenik.
\newblock Statistical topological data analysis using persistence landscapes.
\newblock \emph{The Journal of Machine Learning Research}, 16\penalty0
  (1):\penalty0 77--102, 2015.

\bibitem[Bubenik and Dlotko(2016)]{bubenik2014persistence}
Peter Bubenik and Pawel Dlotko.
\newblock A persistence landscapes toolbox for topological statistics.
\newblock \emph{Jounral of Symbolic Computations}, 2016.
\newblock Accepted.

\bibitem[Carlsson(2009)]{carlsson2009topology}
Gunnar Carlsson.
\newblock Topology and data.
\newblock \emph{Bulletin of the American Mathematical Society}, 46\penalty0
  (2):\penalty0 255--308, 2009.

\bibitem[Carri{\`e}re et~al.(2015)Carri{\`e}re, Oudot, and
  Ovsjanikov]{carriere2015stable}
Mathieu Carri{\`e}re, Steve~Y Oudot, and Maks Ovsjanikov.
\newblock Stable topological signatures for points on 3d shapes.
\newblock In \emph{Computer Graphics Forum}, volume~34, pages 1--12, 2015.

\bibitem[Chazal et~al.(2014)Chazal, de~Silva, and Oudot]{chazal2014persistence}
Fr{\'e}d{\'e}ric Chazal, Vin de~Silva, and Steve Oudot.
\newblock Persistence stability for geometric complexes.
\newblock \emph{Geometriae Dedicata}, 173\penalty0 (1):\penalty0 193--214,
  2014.

\bibitem[Chen et~al.(2015)Chen, Wang, Rinaldo, and
  Wasserman]{chen2015statistical}
Yen-Chi Chen, Daren Wang, Alessandro Rinaldo, and Larry Wasserman.
\newblock Statistical analysis of persistence intensity functions.
\newblock \emph{arXiv preprint arXiv:1510.02502}, 2015.

\bibitem[Chepushtanova et~al.(2014)Chepushtanova, Gittins, and Kirby]{chep}
Sofya Chepushtanova, Christopher Gittins, and Michael Kirby.
\newblock Band selection in hyperspectral imagery using sparse support vector
  machines.
\newblock In \emph{Proceedings SPIE DSS 2014}, volume 9088, pages
  90881F--90881F15, 2014.

\bibitem[Chung et~al.(2009)Chung, Bubenik, and Kim]{chung2009persistence}
Moo~K Chung, Peter Bubenik, and Peter~T Kim.
\newblock Persistence diagrams of cortical surface data.
\newblock In \emph{Information Processing in Medical Imaging}, pages 386--397.
  Springer, 2009.

\bibitem[Cohen-Steiner et~al.(2007)Cohen-Steiner, Edelsbrunner, and
  Harer]{stabilityPD}
David Cohen-Steiner, Herbert Edelsbrunner, and John Harer.
\newblock Stability of persistence diagrams.
\newblock \emph{Discrete \& Computational Geometry}, 37\penalty0 (1):\penalty0
  103--120, 2007.

\bibitem[Cohen-Steiner et~al.(2010)Cohen-Steiner, Edelsbrunner, Harer, and
  Mileyko]{cohen2010lipschitz}
David Cohen-Steiner, Herbert Edelsbrunner, John Harer, and Yuriy Mileyko.
\newblock Lipschitz functions have ${L}_p$-stable persistence.
\newblock \emph{Foundations of computational mathematics}, 10\penalty0
  (2):\penalty0 127--139, 2010.

\bibitem[Cuerno and Barab\'asi(1995)]{cb95}
Rodolfo Cuerno and Albert-L\'asl\'o Barab\'asi.
\newblock Dynamic scaling of ion-sputtered surfaces.
\newblock \emph{Physical Review Letters}, 74:\penalty0 4746, 1995.

\bibitem[Dabaghian et~al.(2012)Dabaghian, Memoli, Frank, and
  Carlsson]{hippocampalPH}
Yu~Dabaghian, Facundo Memoli, L~Frank, and Gunnar Carlsson.
\newblock A topological paradigm for hippocampal spatial map formation using
  persistent homology.
\newblock \emph{PLoS Computational Biology}, 8\penalty0 (8):\penalty0 e1002581,
  2012.

\bibitem[Di~Fabio and Ferri(2015)]{di2015comparing}
Barbara Di~Fabio and Massimo Ferri.
\newblock Comparing persistence diagrams through complex vectors.
\newblock In \emph{International Conference on Image Analysis and Processing
  2015 Part I; Editors V.\ Murino, E.\ Puppo, LNCS 9279}, pages 294--305, 2015.

\bibitem[Donatini et~al.(1998)Donatini, Frosini, and Lovato]{donatini1998size}
Pietro Donatini, Patrizio Frosini, and Alberto Lovato.
\newblock Size functions for signature recognition.
\newblock In \emph{SPIE's International Symposium on Optical Science,
  Engineering, and Instrumentation}, pages 178--183, 1998.

\bibitem[Edelsbrunner and Harer(2010)]{Edelsbrunner10}
Herbert Edelsbrunner and John Harer.
\newblock \emph{Computational topology: {A}n introduction}.
\newblock American Mathematical Society, 2010.

\bibitem[Edelsbrunner and Harer(2008)]{edelsbrunner2008persistent}
Herbert Edelsbrunner and John Harer.
\newblock Persistent homology -- a survey.
\newblock \emph{Contemporary Mathematics}, 453:\penalty0 257--282, 2008.

\bibitem[Edelsbrunner et~al.(2012)Edelsbrunner, Ivanov, and
  Karasev]{edelsbrunner2012current}
Herbert Edelsbrunner, A~Ivanov, and R~Karasev.
\newblock Current open problems in discrete and computational geometry.
\newblock \emph{Modelirovanie i Analiz Informats. Sistem}, 19\penalty0
  (5):\penalty0 5--17, 2012.

\bibitem[Emerson et~al.(2015)Emerson, Kirby, Bethel, Kolatkar, Luttgen, O'Hara,
  Newton, and Kuhn]{emersonCTC}
Tegan Emerson, Michael Kirby, Kelly Bethel, Anand Kolatkar, Madelyn Luttgen,
  Stephen O'Hara, Paul Newton, and Peter Kuhn.
\newblock Fourier-ring descriptor to characterize rare circulating cells from
  images generated using immunofluorescence microscopy.
\newblock \emph{Computerized Medical Imaging and Graphics}, 40:\penalty0
  70--87, 2015.

\bibitem[Ferri and Landi(1999)]{ferri1999representing}
Massimo Ferri and Claudia Landi.
\newblock Representing size functions by complex polynomials.
\newblock \emph{Proc.\ Math.\ Met.\ in Pattern Recognition}, 9:\penalty0
  16--19, 1999.

\bibitem[Ferri et~al.(1997)Ferri, Frosini, Lovato, and
  Zambelli]{ferri1997point}
Massimo Ferri, Patrizio Frosini, Alberto Lovato, and Chiara Zambelli.
\newblock Point selection: {A} new comparison scheme for size functions (with
  an application to monogram recognition).
\newblock In \emph{Computer Vision ACCV'98}, pages 329--337. Springer, 1997.

\bibitem[Fung and Mangasarian(2004)]{Fung2004}
Glenn~M. Fung and O.L. Mangasarian.
\newblock A feature selection newton method for support vector machine
  classification.
\newblock \emph{Computational Optimization and Applications}, 28\penalty0
  (2):\penalty0 185--202, 2004.

\bibitem[Ghrist(2008)]{barcodes}
Robert Ghrist.
\newblock Barcodes: {T}he persistent topology of data.
\newblock \emph{Bulletin of the American Mathematical Society}, 45\penalty0
  (1):\penalty0 61--75, 2008.

\bibitem[Golovin and Davis(1998)]{golovin98}
Alexander~A. Golovin and Stephen~H Davis.
\newblock Effect of anisotropy on morphological instability in the freezing of
  a hypercooled melt.
\newblock \emph{Physica D: Nonlinear Phenomena}, 116:\penalty0 363--391, 1998.

\bibitem[Guo et~al.(2010)Guo, Zhang, and Zhang]{guo2010completed}
Zhenhua Guo, Lei Zhang, and David Zhang.
\newblock A completed modeling of local binary pattern operator for texture
  classification.
\newblock \emph{IEEE Transactions on Image Processing}, 19\penalty0
  (6):\penalty0 1657--1663, 2010.

\bibitem[Hatcher(2002)]{hatcher2002algebraic}
Allen Hatcher.
\newblock \emph{Algebraic Topology}.
\newblock Cambridge University Press, 2002.

\bibitem[Heath et~al.(2010)Heath, Gelfand, Ovsjanikov, Aanjaneya, and
  Guibas]{imagewebs}
Kyle Heath, Natasha Gelfand, Maks Ovsjanikov, Mridul Aanjaneya, and Leonidas~J
  Guibas.
\newblock Image webs: {C}omputing and exploiting connectivity in image
  collections.
\newblock In \emph{Computer Vision and Pattern Recognition (CVPR), 2010 IEEE
  Conference on}, pages 3432--3439. IEEE, 2010.

\bibitem[Hertzsch et~al.(2007)Hertzsch, Sturman, and Wiggins]{DNAHertzsch}
Jan-Martin Hertzsch, Rob Sturman, and Stephen Wiggins.
\newblock {DNA} microarrays: Design principles for maximizing ergodic, chaotic
  mixing.
\newblock \emph{Small}, 3\penalty0 (2):\penalty0 202--218, 2007.

\bibitem[Ho(1998)]{ho1998random}
Tin~Kam Ho.
\newblock The random subspace method for constructing decision forests.
\newblock \emph{IEEE Transactions on Pattern Analysis and Machine
  Intelligence}, 20\penalty0 (8):\penalty0 832--844, 1998.

\bibitem[Kaufman and Rousseeuw(1987)]{kaufman1987clustering}
Leonard Kaufman and Peter Rousseeuw.
\newblock \emph{Clustering by means of medoids}.
\newblock North-Holland, 1987.

\bibitem[Kerber et~al.(2016)Kerber, Morozov, and Nigmetov]{kerbergeometry}
Michael Kerber, Dmitriy Morozov, and Arnur Nigmetov.
\newblock Geometry helps to compare persistence diagrams.
\newblock \emph{Proceedings of the Workshop on Algorithm Engineering and
  Experiments}, 2016.
\newblock Accepted.

\bibitem[Mileyko et~al.(2011)Mileyko, Mukherjee, and Harer]{probabilityonPD}
Yuriy Mileyko, Sayan Mukherjee, and John Harer.
\newblock Probability measures on the space of persistence diagrams.
\newblock \emph{Inverse Problems}, 27\penalty0 (12):\penalty0 124007, 2011.

\bibitem[Motta et~al.(2012)Motta, Shipman, and Bradley]{Motta12}
Francis~C Motta, Patrick~D Shipman, and R~Mark Bradley.
\newblock Highly ordered nanoscale surface ripples produced by ion bombardment
  of binary compounds.
\newblock \emph{Journal of Physics D: Applied Physics}, 45\penalty0
  (12):\penalty0 122001, 2012.

\bibitem[Pachauri et~al.(2011)Pachauri, Hinrichs, Chung, Johnson, and
  Singh]{pachauri2011topology}
Deepti Pachauri, Christian Hinrichs, Moo~K Chung, Sterling~C Johnson, and Vikas
  Singh.
\newblock Topology-based kernels with application to inference problems in
  {A}lzheimer's disease.
\newblock \emph{IEEE Transactions on Medical Imaging}, 30\penalty0
  (10):\penalty0 1760--1770, 2011.

\bibitem[Park and Jun(2009)]{park2009simple}
Hae-Sang Park and Chi-Hyuck Jun.
\newblock A simple and fast algorithm for $k$-medoids clustering.
\newblock \emph{Expert Systems with Applications}, 36\penalty0 (2):\penalty0
  3336--3341, 2009.

\bibitem[Pearson et~al.(2015)Pearson, Bradley, Motta, and Shipman]{ions}
Daniel~A. Pearson, R.~Mark Bradley, Francis~C. Motta, and Patrick~D. Shipman.
\newblock Producing nanodot arrays with improved hexagonal order by patterning
  surfaces before ion sputtering.
\newblock \emph{Phys. Rev. E}, 92:\penalty0 062401, Dec 2015.
\newblock \doi{10.1103/PhysRevE.92.062401}.
\newblock URL \url{http://link.aps.org/doi/10.1103/PhysRevE.92.062401}.

\bibitem[Perea and Harer(2013)]{windowsandpersistence}
Jose~A Perea and John Harer.
\newblock Sliding windows and persistence: An application of topological
  methods to signal analysis.
\newblock \emph{Foundations of Computational Mathematics}, pages 1--40, 2013.

\bibitem[Reininghaus et~al.(2015)Reininghaus, Huber, Bauer, and
  Kwitt]{reininghaus2015stable}
Jan Reininghaus, Stefan Huber, Ulrich Bauer, and Roland Kwitt.
\newblock A stable multi-scale kernel for topological machine learning.
\newblock In \emph{Proceedings of the IEEE Conference on Computer Vision and
  Pattern Recognition}, pages 4741--4748, 2015.

\bibitem[Rost and Krug(1995)]{rost95}
Martin Rost and Joachim Krug.
\newblock Anisotropic kuramoto-sivashinsky equation for surface growth and
  erosion.
\newblock \emph{Physical Review Letters}, 75:\penalty0 3894, 1995.

\bibitem[Singh et~al.(2008)Singh, Memoli, Ishkhanov, Sapiro, Carlsson, and
  Ringach]{visionTDA}
Gurjeet Singh, Facundo Memoli, Tigran Ishkhanov, Guillermo Sapiro, Gunnar
  Carlsson, and Dario~L Ringach.
\newblock Topological analysis of population activity in visual cortex.
\newblock \emph{Journal of Vision}, 8\penalty0 (8):\penalty0 11, 2008.

\bibitem[Topaz et~al.(2015)Topaz, Ziegelmeier, and Halverson]{Swarms}
Chad~M Topaz, Lori Ziegelmeier, and Tom Halverson.
\newblock Topological data analysis of biological aggregation models.
\newblock \emph{PloS One}, 10\penalty0 (5):\penalty0 e0126383, 2015.

\bibitem[Turner et~al.(2014)Turner, Mileyko, Mukherjee, and
  Harer]{turner2014frechet}
Katharine Turner, Yuriy Mileyko, Sayan Mukherjee, and John Harer.
\newblock Fr{\'e}chet means for distributions of persistence diagrams.
\newblock \emph{Discrete \& Computational Geometry}, 52\penalty0 (1):\penalty0
  44--70, 2014.

\bibitem[Verov\v{s}ek(2016)]{verovsek2016tropical}
Sara~Kali\v{s}nik Verov\v{s}ek.
\newblock Tropical coordinates on the space of persistence barcodes.
\newblock \emph{arXiv preprint arXiv:1604.00113}, 2016.

\bibitem[Villain(1991)]{villain91}
J~Villain.
\newblock Continuum models of crystal growth from atomic beams with and without
  desorption.
\newblock \emph{J. Phys. I France}, 1:\penalty0 19--42, 1991.

\bibitem[Wolf(1991)]{wolf91}
Dietrich~E Wolf.
\newblock Kinetic roughening of vicinal surfaces.
\newblock \emph{Physical Review Letters}, 67:\penalty0 1783, 1991.

\bibitem[Zeppelzauer et~al.(2016)Zeppelzauer, Zieli{\'n}ski, Juda, and
  Seidl]{zeppelzauer2016topological}
Matthias Zeppelzauer, Bartosz Zieli{\'n}ski, Mateusz Juda, and Markus Seidl.
\newblock Topological descriptors for 3d surface analysis.
\newblock \emph{arXiv preprint arXiv:1601.06057}, 2016.

\bibitem[Zhang and Zhou(2010)]{Zhang2010b}
Li~Zhang and Weida Zhou.
\newblock On the sparseness of 1-norm support vector machines.
\newblock \emph{Neural Networks}, 23\penalty0 (3):\penalty0 373--385, 2010.

\bibitem[Zhu et~al.(2004)Zhu, Rosset, Hastie, and Tibshirani]{Zhu2003}
Ji~Zhu, Saharon Rosset, Trevor Hastie, and Rob Tibshirani.
\newblock 1-norm support vector machines.
\newblock \emph{Advances in neural information processing systems}, 16\penalty0
  (1):\penalty0 49--56, 2004.

\bibitem[Zomorodian and Carlsson(2005)]{computingPH}
Afra Zomorodian and Gunnar Carlsson.
\newblock Computing persistent homology.
\newblock \emph{Discrete \& Computational Geometry}, 33\penalty0 (2):\penalty0
  249--274, 2005.

\end{thebibliography}

\newpage
\appendix

\section{Homology and Data} \label{app:homology}

Homology is an invariant that characterizes the topological properties of a topological space $X$. In particular, homology measures the number of connected components, loops, trapped volumes, and so on of a topological space, and can be used to distinguish distinct spaces from one another. More explicitly, the $k$-dimensional holes of a space generate a homology group, $H_k(X)$. The rank of this group is referred to as the \emph{$k$-th Betti number}, $\beta_k$, and counts the number of $k$-dimensional holes of $X$. For a comprehensive study of homology, see \citet{hatcher2002algebraic}.

\subsection{Simplicial Complexes and Homology}

\indent Simplicial complexes are one way to define topological spaces combinatorially. More precisely, a \emph{simplicial complex} $S$ consists of vertices (0-simplices), edges (1-simplices), triangles (2-simplices), tetrahedra (3-simplices), and higher-dimensional $k$-simplices (containing $k+1$ vertices), such that
\begin{itemize}
\item if $\sigma$ is a simplex in $S$ then $S$ contains all lower-dimensional simplices of $\sigma$, and
\item the non-empty intersection of any two simplices in $S$ is a simplex in $S$.
\end{itemize}

The following setup is necessary for a rigorous definition of (simplicial) homology. To a simplicial complex, one can associate a chain complex of vector spaces over a field $\mathbb{F}$ (often a finite field $\mathbb{Z}/p\mathbb{Z}$ for $p$ a small prime), 
\[\cdots \rightarrow C_{k+1} \xrightarrow{\partial_{k+1}} C_k \xrightarrow{\partial_{k}} C_{k-1} \rightarrow \cdots.\]
\noindent Here, vector space $C_k$ consists of all $\mathbb{F}$-linear combinations of the $k$ simplices of $S$, and has as a basis the set of all $k$-simplices. The linear map $\partial_k: C_k \rightarrow C_{k-1}$, known as the \emph{boundary operator}, maps a $k$-simplex to its boundary, a sum of its $(k-1)$-faces. More formally, the boundary map acts on a $k$-simplex $[v_0,v_1,\ldots,v_k]$ by
\[\partial_k([v_0,v_1, \ldots, v_k]) = \sum_{i=0}^k (-1)^i [v_0, \ldots, \hat{v_i}, \ldots, v_k],\]
where $[v_0, \ldots, \hat{v_i}, \ldots, v_k]$ is the $(k-1)$-simplex obtained from $[v_0, \ldots, v_k]$ by removing vertex $v_i$. We define two subspaces of $C_k$: subspace $Z_k = \ker(\delta_k)$ is known as the \emph{$k$-cycles}, and subspace $B_k = \im(\delta_{k+1}) = \delta_{k+1}(C_{k+1})$ is known as the \emph{$k$-boundaries}. The boundary operator satisfies the property $\partial_k \circ \partial_{k+1}=0$, which implies the inclusion $B_k\subseteq Z_k$.

Homology seeks to uncover an equivalence class of cycles that enclose a $k$-dimensional hole---that is, cycles which are not also boundaries of $k$-simplices. To this end, the $k$-th order homology is defined as $H_k(S) = Z_k/B_k$, a quotient of vector spaces. The $k$-th Betti number $\beta_k = \dim(H_k(S))$ is the dimension of this vector space, and counts the number of independent holes of dimension $k$. More explicitly, $\beta_0$ counts the number of connected components, $\beta_1$ the number of loops, $\beta_2$ the number of trapped volumes, and so on. Betti numbers are a topological invariant, meaning that topologically equivalent spaces have the same Betti number.

\subsection{Persistence Diagrams from Point Cloud Data}\label{PD_Data}

One way to approximate the topological characteristics of a point cloud dataset is to build a simplicial complex on top of it. Though there are a variety of methods to do so, we restrict attention to the \emph{Vietoris--Rips simplicial complex} due to its computational tractability \citep{barcodes}. Given a data set $Y$ (equipped with a metric) and a scale parameter $\epsilon\ge0$, the Vietoris--Rips complex $S_\epsilon$ has $Y$ as its set of vertices and has a $k$-simplex for every collection of $k+1$ vertices whose pairwise distance is at most $\epsilon$. However, it is often not apparent how to choose scale $\epsilon$. Selecting $\epsilon$ too small results in a topological space with a large number of connected components, and selecting $\epsilon$ too large results in a topological space that is contractible (equivalent to a single point).

The idea of persistent homology is to compute homology at many scales and observe which topological features persist across those scales \citep{barcodes,carlsson2009topology,edelsbrunner2008persistent}. Indeed, if $\epsilon_1 \leq \epsilon_2 \leq \ldots \leq \epsilon_m$ is an increasing sequence of scales, then the corresponding Vietoris--Rips simplicial complexes form a filtered sequence $S_{\epsilon_1} \subseteq S_{\epsilon_2} \subseteq \ldots \subseteq S_{\epsilon_m}$. As $\epsilon$ varies, so does the homology of $S_\epsilon$, and for any homological dimension $k$ we get a sequence of linear maps $H_k(S_{\epsilon_1}) \to H_k(S_{\epsilon_2}) \to \ldots \to H_k(S_{\epsilon_m})$. Persistent homology tracks the homological features over a range of values of $\epsilon$. Those features which persist over a larger range are considered to be true topological characteristics, while short-lived features are often considered as noise.

For each choice of homological dimension $k$, the information measured by persistent homology can be presented as a \emph{persistence diagram} (PD), a multiset of points in the plane. Each point $(x,y)=(\epsilon,\epsilon')$ corresponds to a topological feature that appears (is `born') at scale parameter $\epsilon$ and which no longer remains (`dies') at scale $\epsilon'$. Since all topological features die after they are born, this is an embedding into the upper half plane, above the diagonal line $y=x$. Points near the diagonal are considered to be noise while those further from the diagonal represent more robust topological features.

\subsection{Persistence Diagrams from Functions}\label{app:PD_Functions}

Let $X$ be a topological space and let $f\colon X\to\bR$ be a real-valued function. One way to understand the behavior of map $f$ is to understand the topology of its sublevel sets $f^{-1}((-\infty,\epsilon])$, where $\epsilon\in\bR$. Indeed, given $\epsilon_1 \leq \epsilon_2 \leq \ldots \leq \epsilon_m$, one can study map $f$ using the persistent homology of the resulting filtration of topological spaces, known as the sublevel set filtration:
\[ f^{-1}((-\infty,\epsilon_1]) \subseteq f^{-1}((-\infty,\epsilon_2]) \subseteq \ldots \subseteq f^{-1}((-\infty,\epsilon_m]). \]
If $X$ is a simplicial complex, then one can produce an increasing sequence of simplicial complexes using a modification of this procedure called the lower star filtration \citep{Edelsbrunner10}. Similarly, if $X$ is a cubical complex (an analogue of a simplicial complex that is instead a union of vertices, edges, squares, cubes, and higher-dimensional cubes), then one can produce an increasing sequence of cubical complexes.

In \S\ref{sec:aks}, we study surfaces $u\colon[0,1]^2\to\bR$ produced from the Kuramoto-Sivashinsky equation. The domain $[0,1]^2$ is discretized into a grid of $512\times512$ vertices, i.e.\ a 2-dimensional cubical complex with $512^2$ vertices, $511\cdot512$ horizontal edges, $511\cdot512$ vertical edges, and $511^2$ squares. We produce an increasing sequence of cubical complexes as follows:
\begin{itemize}
\item A vertex $v$ is included at scale $\epsilon$ if $u(v)\le\epsilon$.
\item An edge is included at scale $\epsilon$ if both of its vertices are present.
\item A square is included at scale $\epsilon$ if all four of its vertices are present.
\end{itemize}
Our PDs are obtained by taking the persistent homology of this cubical complex sublevel set filtration.

We remark that PDs from point cloud data in \S\ref{PD_Data} can be viewed as a specific case of PDs from functions. Indeed, given a data set $X$ in some metric space $(M,d)$, let $d_X\colon M\to\bR$ be the distance function to set $X$, defined by $d_X(m)=\inf_{x\in X}d(x,m)$ for all $m\in M$. Note that $d_X^{-1}((-\infty,\epsilon])$ is the union of the metric balls of radius $\epsilon$ centered at each point in $X$. For $\epsilon_1 \leq \epsilon_2 \leq \ldots \leq \epsilon_m$, the persistent homology of

\[ d_X^{-1}((-\infty,\epsilon_1]) \subseteq d_X^{-1}((-\infty,\epsilon_2]) \subseteq \ldots \subseteq d_X^{-1}((-\infty,\epsilon_m]) \]
is identical to the persistent homology of a simplicial complex filtration called the \emph{\v Cech complex}. Furthermore, the persistent homology of the Vietoris--Rips complex is an approximation of the persistent homology of the \v Cech complex \citep[Section~III.2]{Edelsbrunner10}.

\section{Examples of Persistence Images}\label{app:PIEx}

\begin{figure}[H]
\centering
\includegraphics[width=.9\textwidth]{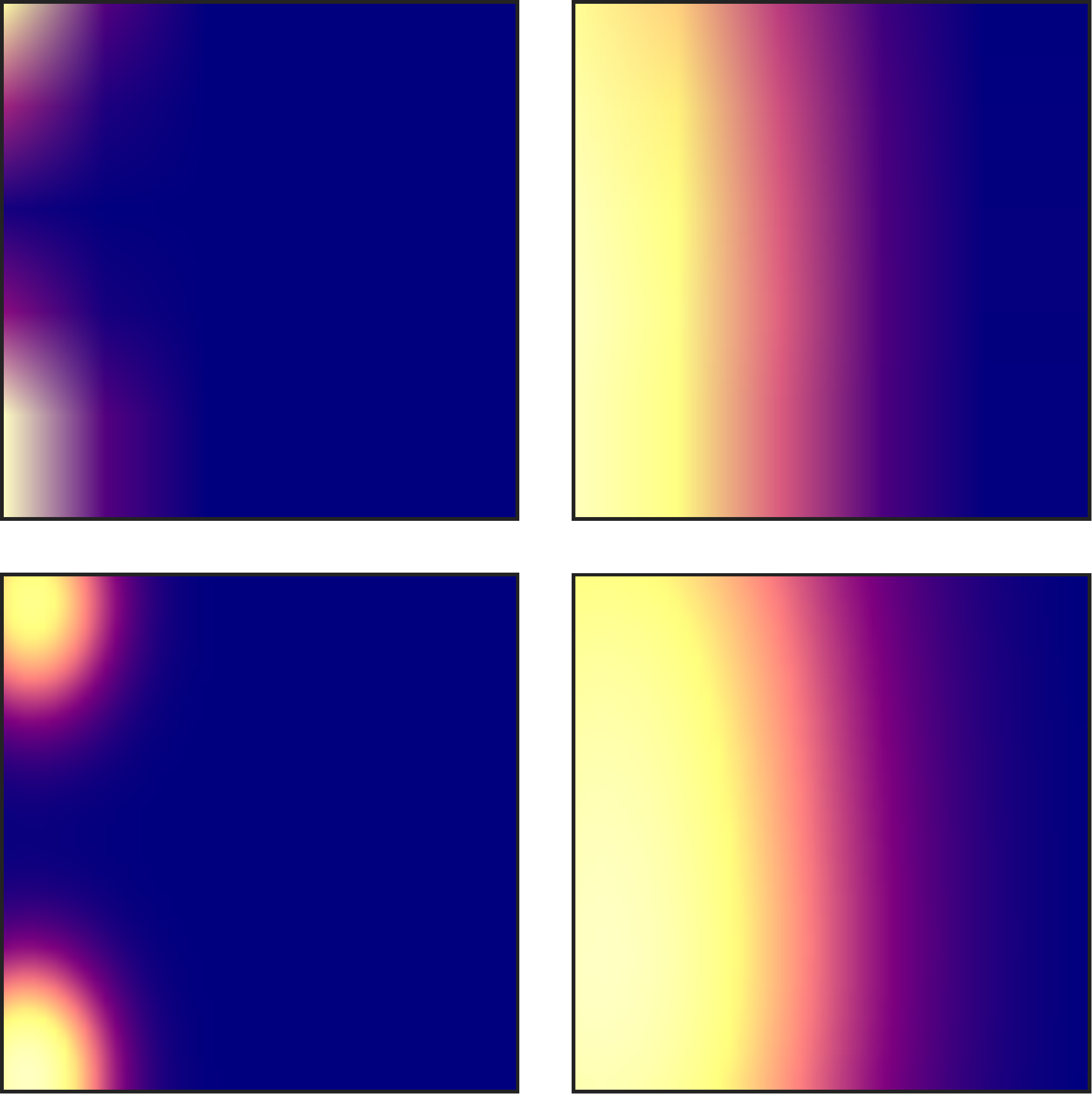}
\caption{Examples of PIs for homology dimension $H_1$ arising from a noisy circle with a variety of resolutions and variances. The first row has resolution $5\times 5$ while the second has $50\times 50$. The columns have variance $\sigma=0.01$ and $\sigma=0.2$, respectively.
}
\label{fig:PIsCircle}
\end{figure}

\begin{figure}[H]
\centering
\includegraphics[width=.9\textwidth]{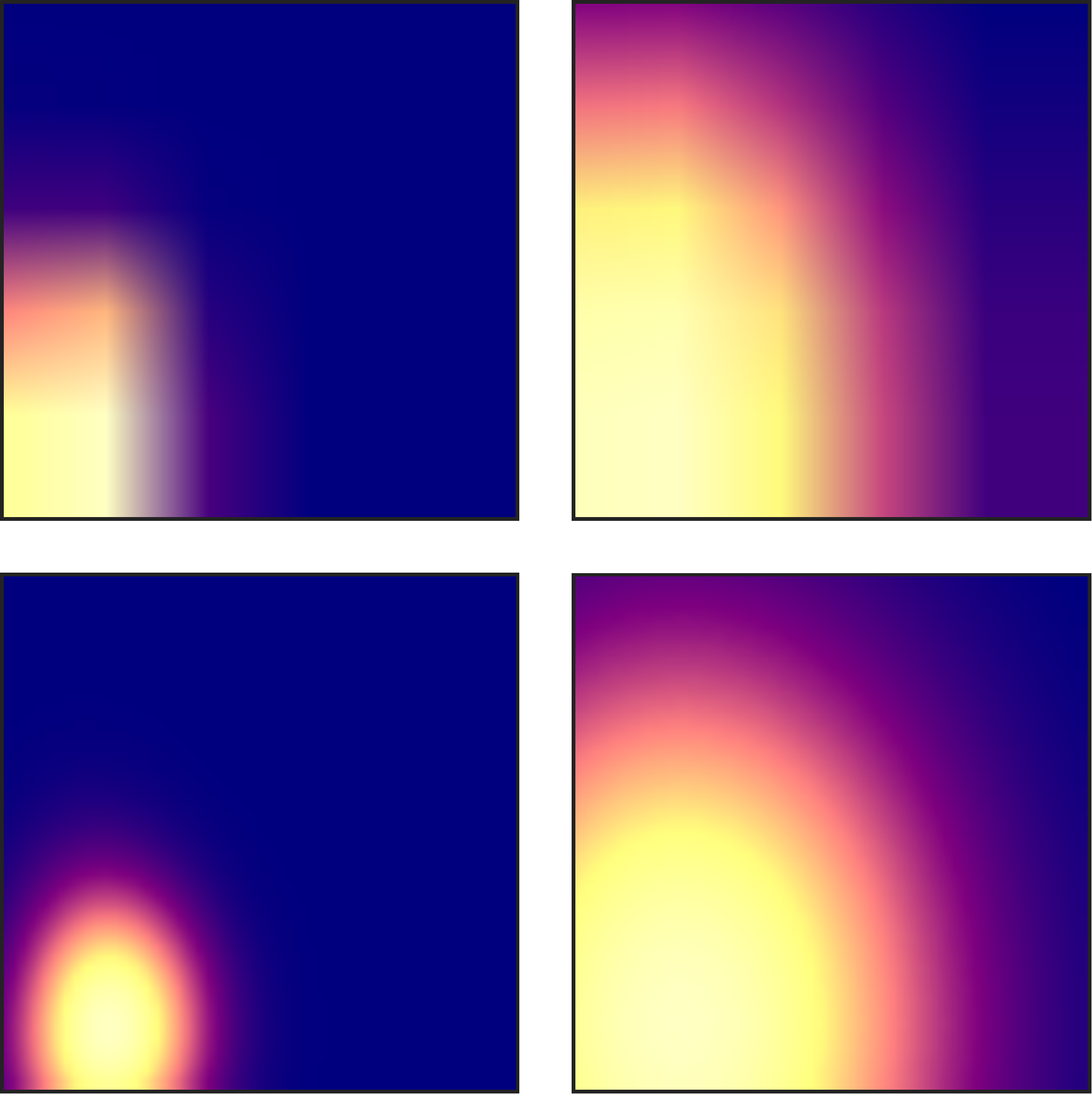}
\caption{Examples of PIs for homology dimension $H_1$ arising from a noisy torus with a variety of resolutions and variances. The first row has resolution $5\times 5$ while the second has $50\times 50$. The columns have variance $\sigma=0.01$ and $\sigma=0.2$, respectively.}
\label{fig:PIsTorus}
\end{figure}

\section{Proofs of Equation~\eqref{eq:Erf} and Lemma~\ref{lem:l1surfbound_2D}} \label{app:proofs}

Let $u,v\in\bR$ and $a,b>0$. Equation~\eqref{eq:Erf} states that $\|ag_u - b g_v\|_1 = F(v-u)$, where $F\colon\bR\to\bR$ is defined by
\[ F(z) = \begin{cases}
|a-b| &\mbox{if }z=0\\
\left|a \text{Erf}\left(\frac{z^2+2\sigma^2\ln(a/b)}{z \sigma 2\sqrt{2}}\right)-b\text{Erf}\left(\frac{-z^2+2\sigma^2\ln(a/b)}{z \sigma 2\sqrt{2}}\right)\right| & \mbox{otherwise.}
\end{cases} \]

\begin{proof}
If $v=u$ then the statement follows from the fact that $g_u$ and $g_v$ are normalized to have unit area under the curve. Hence we may assume $u\neq v$.

For $u \neq v$ a straightforward calculation shows there is a unique solution $z^*$ to $a g_u(z) = b g_v (z)$,
namely
\[z^* = \dfrac{v^2 - u^2 + 2\sigma^2 \ln(a/b)}{2(v-u)}.\]
Note
\begin{equation}\label{eq:agu-bgv}
\|ag_u-bg_v\|_1 = \int_{-\infty}^{\infty} |ag_u(z) - bg_v(z)| dz = \left| \int_{-\infty}^{z^*} ag_u(z) - bg_v(z) dz + \int_{z^*}^{\infty} bg_v(z) - ag_u(z) dz\right|.
\end{equation}
There are four integrals to compute, and we do each one in turn. We have
\begin{align*}
\int_{-\infty}^{z^*} ag_u(z)\ dz & = \frac{a}{\sigma\sqrt{2\pi}}\int_{-\infty}^{z^*} e^{-(z-u)^2/2\sigma^2}\ dz \\
& = \frac{a}{\sqrt{\pi}}\int_{-\infty}^{P(v-u)} e^{-t^2}\ dt &&\mbox{by substitution $t = \frac{z-u}{\sigma\sqrt{2}}$}\\
& = \frac{a}{\sqrt{\pi}} \left(\int_{-\infty}^{0} e^{-t^2}\ dt + \int_{0}^{P(v-u)} e^{-t^2}\ dt \right) \\
& =  \frac{a}{\sqrt{\pi}} \left(\frac{\sqrt{\pi}}{2} + \frac{\sqrt{\pi}}{2} \text{Erf}(P(v-u)) \right) \\
& = \frac{a}{2}\left( 1 + \text{Erf}(P(v-u))\right),
\end{align*}
where $P(z) = \frac{z^2+2\sigma^2\ln(a/b)}{z \sigma 2\sqrt{2}}$. Nearly identical calculations show
\begin{align*}
\int_{z^*}^{\infty} a g_u(z)\ dz & =  \frac{a}{2}\left( 1 - \text{Erf}(P(v-u))\right)\\	
\int_{-\infty}^{z^*} b g_v(z)\ dz & =  \frac{b}{2}\left( 1 + \text{Erf}(Q(v-u))\right)\\
\int_{z^*}^{\infty} b g_v(z)\ dz & =  \frac{b}{2}\left( 1 - \text{Erf}(Q(v-u))\right),
\end{align*}
where $Q(z) = \frac{-z^2+2\sigma^2\ln(a/b)}{z \sigma 2\sqrt{2}}$.
Plugging back into \eqref{eq:agu-bgv} gives $\|ag_u - b g_v\|_1 = F(v-u)$.
\end{proof}

We now give the proof of Lemma~\ref{lem:l1surfbound_2D}.\\

\noindent\textbf{Lemma~\ref{lem:l1surfbound_2D}.}
 For $u,v\in\bR^2$, let $g_u,g_v \colon \bR^2 \to \bR$ be normalized 2-dimensional Gaussians. Then\[ \displaystyle  \|f(u)g_u - f(v)g_v\|_1 \le \left(|\nabla f|+\sqrt{\frac{2}{\pi}}\frac{\min\{f(u),f(v)\}}{\sigma}\right)\|u-v\|_2. \]
 
\begin{proof}
The result will follow from the observation that we can reduce the two-dimensional case involving Gaussians centered at $u,v \in \mathbb{R}^2$ to one-dimensional Gaussians centered at 0 and $r = \|u-v\|_2$. Let $u = (u_x,u_y)$ and $v = (v_x,v_y)$; we may assume $u_x > v_x$ w.l.o.g. Let $(r,\theta)$ be the magnitude and angle of vector $u-v$ when expressed in polar coordinates. The change of variables $(z,w)=R_{\theta}(x-v_x,y-v_y)$, where $R_{\theta}$ is the clockwise rotation of the plane by $\theta$, gives
\begin{align*}
&\|f(u)g_u - f(v)g_v\|_1 \\
=&  \int_{-\infty}^{\infty} \int_{-\infty}^{\infty} \left| \frac{f(u)}{2\pi \sigma^2} e^{-[(x-u_x)^2+(y-u_y)^2]/2\sigma^2} - \frac{f(v)}{2\pi \sigma^2} e^{-[(x-v_x)^2+(y-v_y)^2]/2\sigma^2} \right| dy \; dx \\
=& \int_{-\infty}^{\infty} \int_{-\infty}^{\infty} \left|\frac{f(u)}{2\pi \sigma^2} e^{-[w^2+(z-r)^2]/2\sigma^2} - \frac{f(v)}{2\pi \sigma^2} e^{-[w^2 + z^2]/2\sigma^2} \right| dz \; dw \\
=& \int_{-\infty}^{\infty} \frac{1}{\sigma \sqrt{2\pi}} e^{-w^2/2\sigma^2} \left[ \int_{-\infty}^{\infty}\left|\frac{f(u)}{\sigma\sqrt{2\pi}} e^{-(z-r)^2/2\sigma^2} - \frac{f(v)}{\sigma\sqrt{2\pi}}e^{-z^2/2\sigma^2} \right| dz \right] dw \\
=& \|f(u)g_r-f(v)g_0\|_1 \int_{-\infty}^{\infty} \frac{1}{\sigma \sqrt{2\pi}} e^{-w^2/2\sigma^2} dw \quad \mbox{with }g_0,g_r\mbox{ 1-dimensional Gaussians} \\
=& \|f(u)g_r-f(v)g_0\|_1 \\
\le& |f(u)-f(v)|+\sqrt{\frac{2}{\pi}}\frac{\min\{f(u),f(v)\}}{\sigma}\|u-v\|_2 \quad \mbox{by Lemma~\ref{lem:l1surfbound_1D}} \\
\le& \left(|\nabla f|+\sqrt{\frac{2}{\pi}}\frac{\min\{f(u),f(v)\}}{\sigma}\right)\|u-v\|_2.
\end{align*}
\end{proof}

\section{SSVM-based Feature Selection} 
\label{app:featureselection}

We performed feature selection using one-against-all (OAA) SSVM on the six classes of synthetic data with noise level $\eta=0.05$. The PIs used in the experiments were generated from the $H_1$ PDs, with the parameter choices of resolution $ 20 \times 20$ and variance $\sigma=0.0001$. Note that because of the resolution parameter choice of $20 \times 20$, each PI is a vector in $\bR^{400}$, and the selected features will be a subset of indices corresponding to pixels within the PI.  We trained an OAA SSVM model for PIs of dimension $H_1$. In the experiment, we used 5-fold cross-validation and obtained $100\%$ overall accuracy. Feature selection was performed by retaining the features with non-zero SSVM weights, determined by magnitude comparison using weight ratios \citep{chep}.  The resulting six sparse models contain subsets of discriminatory features for each class. Note that one can use only these selected features for classification without loss in accuracy. These features correspond to discriminatory pixels in the persistence images.

Figure~\ref{fig:FeatureSelection2} shows locations of pixels in the vectorized PIs selected by OAA SSVM that discriminate each class from all the others. 

\begin{figure}[H]
\centering
\includegraphics[width=0.9\textwidth]{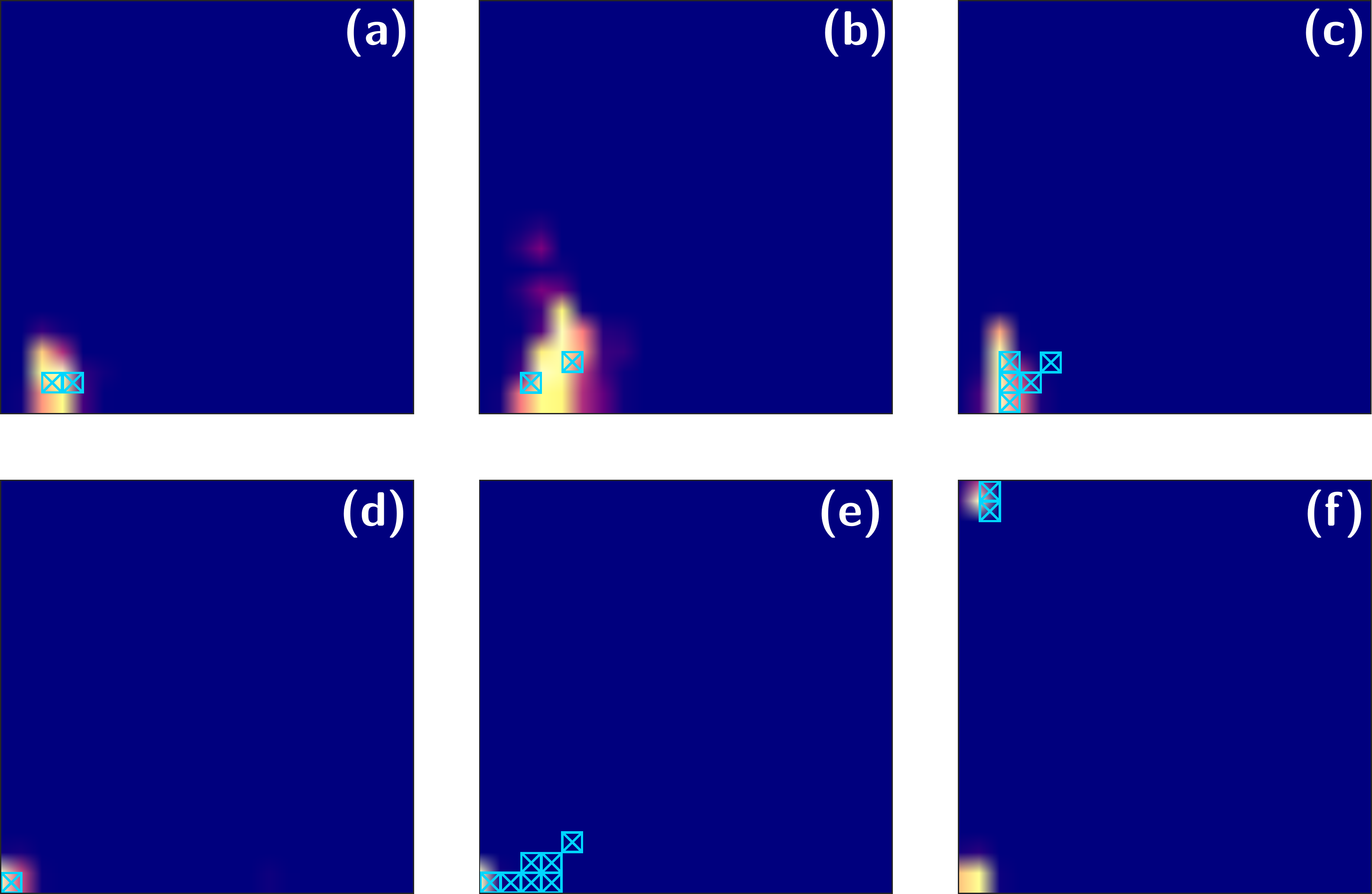}
\caption{SSVM-based feature (pixel) selection for $H_1$ PIs from the six classes of the synthetic data. The parameters used are resolution $20 \times 20$ and variance $0.0001$, for noise level $0.05$. Selected pixels are marked by blue crosses. (a) A noisy solid cube with the two selected pixels (indices 59 and 79 out of 400). (b) A noisy torus with the two selected pixels (indices 59 and 98 out of 400). (c) A noisy sphere with the five selected pixels (indices 58, 59, 60, 79, and 98 out of 400). (d) Noisy three clusters with the one selected pixel (index 20 out of 400). (e) Noisy three clusters within three clusters with the seven selected pixels (indices 20, 40, 59, 60, 79, 80, and 98 out of 400). (f) A noisy circle with the two selected pixels (indices 21 and 22 out of 400). }
\label{fig:FeatureSelection2}
\end{figure}


\end{document}